\documentclass[letterpaper, 11pt, onecolumn, draftclsnofoot]{IEEEtran}
\usepackage{fancyhdr}
\usepackage{amsmath,epsfig}
\usepackage{threeparttable}
\usepackage{epsf,epsfig}
\usepackage{amsthm}
\usepackage{amsmath}
\usepackage{amssymb}
\usepackage{amsfonts}
\usepackage{algorithmic}
\usepackage[noadjust]{cite}
\usepackage{dsfont}
\usepackage{color}
\usepackage{subfigure}
\usepackage[ruled,vlined]{algorithm2e}
\newtheorem{lemma}{Lemma}

\newtheorem{proposition}{Proposition}

\def\phi{\varphi}

\def\l{\left}
\def\r{\right}
\def\({\left(}
\def\){\right)}

\newcommand{\eref}[1]{(\ref{#1})}

\def\b0{{\mathbf{0}}}

\def\bH{{\mathbf{H}}}

\begin{document}

\title{\huge \setlength{\baselineskip}{30pt} Feedback-Topology Designs for Interference Alignment in MIMO Interference Channels}
\author{\large \setlength{\baselineskip}{15pt}Sungyoon Cho, Kaibin Huang, Dongku Kim, Vincent K. N.  Lau, Hyukjin Chae, Hanbyul Seo and Byounghoon Kim\thanks{
S. Cho, K. Huang, D. Kim and H. Chae are with the school of EEE, Yonsei University, Korea;  V. K. N.  Lau  is with the dept. of ECE,  Hong Kong University of Science and Technology, Hong Kong; H. Seo and B. Kim are with the Advanced Communication Technology Lab, LG Electronics, Korea. Corresponding author: K. Huang (email: huangkb@ieee.org). 
}}
\maketitle

\begin{abstract}
 Interference alignment (IA) is a joint-transmission technique that achieves the maximum degrees-of-freedom (DoF) of the interference channel, which provides  linear scaling of the capacity with the number of users for  high signal-to-noise ratios (SNRs). Most prior work on IA is  based on the impractical assumption that perfect and global channel-state information (CSI) is available at all transmitters. To implement IA,  each receiver has to feed back CSI to all interferers, resulting in overwhelming feedback overhead. In particular,  the sum feedback rate of each receiver scales quadratically with the number of users even if the quantized CSI is fed back. To substantially suppress feedback overhead, this paper focuses on designing efficient arrangements of feedback links, called \emph{feedback topologies}, under the IA constraint. For the multiple-input-multiple-output (MIMO) $K$-user interference channel, we propose the feedback topology that supports sequential CSI exchange (feedback and feedforward) between transmitters and receivers so as to achieve IA progressively. This feedback topology is shown to reduce the network feedback overhead from a quadratic function of $K$ to a linear one. To reduce the delay in the sequential CSI exchange, an alternative feedback topology is designed for supporting two-hop feedback via a control station, which also achieves the linear feedback scaling with $K$.  Next, given the proposed feedback topologies, the feedback-bit allocation algorithm is designed for allocating feedback bits by each receiver to different feedback links so as to regulate the residual interference caused by the finite-rate feedback. Simulation results demonstrate that the proposed bit allocation leads to significant throughput gains especially in strong interference environments.     
\end{abstract}
%%%%%%%%%%%%%%%%%%%%%%%%%%%%%%%%%%%%%%%%%%%%%%%%%%%%

\begin{keywords}
Interference alignment, degrees of freedom, MIMO, interference channels, feedback topology, limited feedback, feedback-bits allocation 
\end{keywords}
%%%%%%%%%%%%%%%%%%%%%%%%%%%%%%%%%%%%%%%%%%%%%%%%%%%

\section{Introduction}

In a wireless interference network, interference alignment (IA) maximizes the number of decoupled data links, called \emph{degrees of freedom} (DoF),  by aligning the cross-link interference signals for each user in a subspace of the signal space extended over time, frequency or space. Such alignment  requires the acquisition of perfect and global channel sate information at transmitters (CSIT), incurring potentially overwhelming CSI feedback overhead in practice. Therefore, the efficient CSIT acquisition remains the key challenge for implementing IA techniques and is the main theme of this paper. Specifically, efficient arrangements of CSI feedback links, called \emph{feedback topologies}, are proposed for reducing  the sum feedback overhead for IA. This overhead is further reduced by dynamically distributing  CSI bits over feedback links under a sum feedback constraint. 

The original IA techniques achieve the maximum DoF of the $K$-user  single-antenna  interference channel, namely $K/2$,  by asymptotic signal-space expansion to attain the ergodicity of the channel variation in time or frequency, called \emph{symbol extension} \cite{CadJafar:InterfAlignment:2007, Shen:Improved_IA:2008, Chung:BeamformingIA:2009 }. Given symbol extension, the bounds on the  achievable DoF for multiple-input-multiple-output (MIMO) interference channel were derived in \cite{gou:BoundMIMOIA:2008, Ghasemi:BoundMIMOIA:2009} and 
the optimal  IA solutions were obtained in closed-form for some specific settings. Due to the impracticality of symbol-extension, recent IA research has been focusing on quantifying the achievable DoF and designing matching IA solutions for a single realization of the MIMO interference channel, called the \emph{MIMO constant channel} \cite{Yetis:Feasibility:2009,Negro:Feasibility:2009,Tse:feasiblity:2011, Peters:iterativeMIMO:2009, Gomadam:iterativeMIMO:2008 }. In particular, the IA feasibility conditions were derived in \cite{Yetis:Feasibility:2009, Negro:Feasibility:2009,Tse:feasiblity:2011}  and iterative IA algorithms for achieving such conditions were proposed  in \cite{Peters:iterativeMIMO:2009, Gomadam:iterativeMIMO:2008}, which exploit  the channel reciprocity to achieve distributive implementation. In addition, the IA principle has been extended to design multi-cell precoding  for celluar networks \cite{Tse:Celluar:2009, Tse:Celluar:2010, Lee:Celluar:2010,Tresch:ImperfectCSI:2009 }.

In practice,  CSIT required for IA usually has to rely on finite-rate CSI feedback, called \emph{limited feedback},  from receivers to their interferers, resulting in imperfect CSIT. The required  scaling laws of the number of feedback bits per user for the IA algorithms to achieve the maximum DoF have been derived in \cite{Thukral:LimitedSISO:2009, Varanasi:LimitedMIMO:2009}. In the literature  of  limited feedback, comprehensive limited feedback algorithms have been designed for the single-user  (see e.g.,  \cite{Love:Grassmannian:2003, Mukkavilli:RVQ:2003, Love:Orthogonal:2005, Love:Unitary:2005, Yeung:RVQ:2007}) and multi-user MIMO systems (see e.g., \cite{Jindal:Broadcasting:2006, Yoo:Downlink:2007, Jindal:BD:2008}). However, there are few practical algorithms for limited feedback targeting  IA, which motivates the current work. 

This paper considers the $K$-user constant MIMO interference channel where each transmitter/receiver employs $M$ antennas. Based on the closed-form solution of IA precoders, we propose the feedback topologies which can be implemented by a finite-rate CSI feedforward and feedback links. The contributions of this paper are summarized as follows.
 \begin{itemize}

 \item[1)]  We propose  the centralized-feedback topologies, called as \emph{centralized-receiver feedback} and \emph{star feedback topology}, where a particular receiver or CSI control  station collects CSI from all receivers, computes the IA precoders,  and then communicates them to the transmitters. In the proposed feedback design, the total number of complex coefficients for CSI exchange, referred to as \emph{CSI overhead}, is shown to scale with the number of users $K$ \emph{linearly} rather than \emph{quadratically} for the conventional approach where each receiver feeds back the CSI to all transmitters through the feedback links for the computation of IA precoder \cite{CadJafar:InterfAlignment:2007}-\cite{Ghasemi:BoundMIMOIA:2009}.   
 \item[2)] While the centralized-feedback method is efficient for the reduction in the CSI overhead, it still requires a large amount of CSI overhead between receivers and an additional CSI control station. To address this issue, we further propose \emph{CSI-exchange  feedback topology} where the IA precoders are sequentially computed based on the exchange of pre-determined precoders (under the existence of feedforward/feedback channels) between subsets of transmitters and receivers. The proposed feedback design is performed on the distributed network without the centralized station that gathers CSI from all receivers. As a result, the proposed CSI-exchange feedback topology yields dramatic reduction of CSI overhead especially when $K$ is large.

 \item[3)] For practical implementations, we consider the impact of limited feedback on the performance of the feedback topology in the interference network. Assuming that random vector quantization (RVQ) in \cite{Yeung:RVQ:2007} is used for quantizing CSI, the expected cross-link interference power at each receiver is upper-bounded by sum of exponential functions of the numbers of feedback bits sent by the receiver. Both the centralized-feedback and CSI-exchange topologies are considered in the analysis.   
 
\item[4)] Minimizing the upper bounds on the above interference power gives a dynamic feedback-bit allocation algorithm based on the water-filling principle. Such an algorithm is shown to provide significant capacity gains over the uniform feedback-bit allocation especially for high SNR's. 
 
  \item[5)] Using the proposed feedback topologies, we derive the required number of feedback bits sent by each receiver for achieving the same DoF as the case of perfect CSIT, which increases linearly with $K$ and logarithmically with the transmission power. 
  
  \end{itemize}
The remainder of this  paper is organized as follows. In Section II, the system model is described. The three CSI feedback topologies are proposed in Section III. The effect of CSI-feedback quantization is analyzed and the dynamic feedback-bit allocation algorithm is proposed in Section IV and V, respectively.  Section VI provides simulation results and the concluding remarks are followed in  Section VII.  

%%%%%%%%%%%%%%%%%%%%%%%%%%%%%%%%%%%%%%%%%%%%%%%%%%%%%

\section{System Model}
 
We consider $K$ pairs MIMO interference channel where each node has $M$ antennas and delivers $d$ data streams to the target receiver over a common spectrum. The wireless channels are characterized by path-loss and small-scale fading and all channel-fading coefficients are assumed to be independent and identically distributed (i.i.d) circularly symmetric complex Gaussian random variables with zero mean and unit variance, denoted as $\mathcal{CN}(0,1)$. Let the $M\times M$ matrix $ {\bf{H}}^{[kj]}$ group the fading coefficients of the channel from transmitter $j$  to receiver $k$ and thus $ {\bf{H}}^{[kj]}$ comprises i.i.d. $\mathcal{CN}(0, 1)$ elements. Then, the channel from transmitter $j$  to receiver $k$ can be readily  written as  
$d_{kj}^{ - \alpha /2}{{\bf{H}}^{[kj]}}$, where $\alpha$ is the path-loss exponent and $d_{kj}$ is the propagation distance.  Let ${\mathbf{V}}^{[j]}=\left[ {\bf{v}}_{1}^{[j]} \cdots    {\bf{v}}_{d}^{[j]} \right]$ and  ${\mathbf{R}}^{[k]}=\left[ {\bf{r}}_{1}^{[k]} \cdots    {\bf{r}}_{d}^{[k]} \right]$   denote $M\times d$ precoder at transmitter $j$ and receive filter  at receiver $k$, where ${\left\| {{{\bf{v}}^{[j]}_i}} \right\|^2} = {\left\| {{{\bf{r}}^{[k]}_i}} \right\|^2} = 1$, $\forall i$. Then, the signal vector received  at receiver $k$ for the $i$-th data stream can be written  as
\begin{equation}\label{Eq:IA:Y}
{\bf{y}}^{[k]}_{i}  = \sqrt{P \over d}d_{kk}^{ - \alpha /2}{\bf{H}}^{[k\,k]} {\bf{v}}^{[k]}_{i} {s}^k_i  +  \sum\limits_{l \ne i}\sqrt{P \over d}d_{kk}^{ - \alpha /2}{\bf{H}}^{[k\,k]} {\bf{v}}^{[k]}_{l} {s}^k_l + \sum\limits_{j \ne k} \sqrt{P\over d}d_{kj}^{ - \alpha /2}{{\bf{H}}^{[k\,j]} {\bf{V}}^{[j]} {\bf s}^j  + {\bf{n}}_k } 
\end{equation}
where ${\bf{s}}^k = \left [s{^k_1} \cdots s{^k_d}\right]^{T}$ denotes the data symbols with ${s}^k_i = \mathcal{CN}(0, 1) $, $P$ is the transmission power  and ${\bf{n}}_k$ is additive white Gaussian noise (AWGN) vector with the covariance matrix ${{\bf{I}}_M}$. 

Throughout this paper, we consider FDD system and assume that each receiver has perfect knowledge of the fading coefficients $\l\{\bH^{[km]}\r\}_{m=1}^K$.  For the case of perfect CSIT, all interfering signals at each receiver can be fully eliminated by using IA precoders and ZF receive filters so that the achievable throughput is given by

\begin{equation}\label{Eq:IA:R_P}
R_{\textsf{perfect}} = \sum\limits_{k = 1}^K \sum\limits_{i = 1}^d {{{\log }_2}} \left( {1 + {{{P}{d_{kk}^{-\alpha} }{{\left| {{{{\bf{r}}}^{[k]\dag }_i}{{\bf{H}}^{[kk]}}{{{\bf{ v}}}^{[k]}_i}} \right|}^2}}}{{ }}} \right). 
\end{equation}

\subsection{Closed-form IA Precoder }
In \cite{Tse:feasiblity:2011}, Bresler \emph{et al.} prove that IA over MIMO constant channel is feasible if and only if the number of antennas satisfies $M \ge d(K+1)/2$ under the symmetric square case where all transmitters and receivers are equipped with the same number of antennas. { Moreover, the achievable DoF and feasibility of IA in asymmetric transmit-receive antennas have been studied in $3$-user interference channels \cite{Wang:IA:2011}.}  However, the transceiver designs of IA satisfying above feasibility condition are not explicitly addressed except for $K=3$ and global CSI is required at all transmit sides for the computation of IA precoders \cite{CadJafar:InterfAlignment:2007}. In  \cite{Tresch:PrecoderMIMO:2009}, the closed-form IA solution for a single data transmission has been proposed under the constraint of $K=M+1$. The main principle of closed-form IA is that the $(k+1)$-th and $(k+2)$-th IA precoders are designed for aligning the interfering signals from transmitter $(k+1)$ and $(k+2)$ in the same subspace at receiver $k$. Then, $(K-1)$ dimensional interference vectors lie in $K-2=M-1$ dimensional subspace at each receiver, which allows one dimensional interference-free link for each receiver.  Extending the closed-form IA solution for a single data stream to the case of multiple data streams transmission, we obtain the IA conditions as follows.
\begin{equation}\label{Eq:IA:ex_cond}
\begin{array}{l}
 \,\,\,\,\,\,\,\,\,\,\,\,\,\,\,\,\, \mathsf{span}\left( {{\bf{H}}_{}^{[12]}{{\bf{V}}^{[2]}}} \right)= \mathsf{span}\left( {{\bf{H}}_{}^{[13]}{{\bf{V}}^{[3]}}} \right)\,\,\,\,\,\,\,\,\,\,\,\,\,\,\,{\rm{at}}\,\,{\rm{receiver}}\,{\rm{1}} \\ 
\,\,\,\,\,\,\,\,\,\,\,\,\,\,\,\,\, \mathsf{span}\left( {{\bf{H}}_{}^{[23]}{{\bf{V}}^{[3]}}} \right)= \mathsf{span}\left( {{\bf{H}}_{}^{[24]}{{\bf{V}}^{[4]}}} \right)\,\,\,\,\,\,\,\,\,\,\,\,\,\,\,{\rm{at}}\,\,{\rm{receiver}}\,2 \\ 
\,\,\,\,\,\,\,\,\,\,\,\,\,\,\,\,\,\,\,\,\,\,\,\,\,\,\,\,\,\,\,\,\,\,\,\,\,\,\,\,\,\,\,\,\,\,\,\,\, \,\,\,\,\,\,\,\,\,\,\,\,\vdots  \\ 
\mathsf{span}\left( {{\bf{H}}_{}^{[(K-1)\, K]}{{\bf{V}}^{[K]}}} \right)= \mathsf{span}\left( {{\bf{H}}_{}^{[(K-1)\,1]}{{\bf{V}}^{[1]}}} \right)\,\,\,\,{\rm{at}}\,\,{\rm{receiver}}\,(K-1) \\ 
\,\,\,\,\,\,\,\,\,\,\,\,\,\, \mathsf{span}\left( {{\bf{H}}_{}^{[K1]}{{\bf{V}}^{[1]}}} \right) = \mathsf{span}\left( {{\bf{H}}_{}^{[K2]}{{\bf{V}}^{[2]}}} \right)\,\,\,\,\,\,\,\,\,\,\,\,\,\,{\rm{at}}\,\,{\rm{receiver}}\,K \\ 
 \end{array}
\end{equation}
where $\mathsf{span(\bf{A})}$ denotes the vector space that spanned by the columns of $\mathsf{\bf{A}}$. Note that ${\mathsf{span}}\left( {{{\bf{H}}^{[k-1\,k]}}{{\bf{V}}^{[k]}}} \right) = {\mathsf{span}}\left( {{{\bf{H}}^{[k-1\,k+1]}}{{\bf{V}}^{[k+1]}}} \right)$ $\Rightarrow$ ${\mathsf{span}}\left( {{{\bf{V}}^{[k+1]}}} \right) = {\mathsf{span}}\left( {{{\left( {{{\bf{H}}^{[k-1\,k+1]}}} \right)}^{ - 1}}{{\bf{H}}^{[k-1\,k]}}{{\bf{V}}^{[k]}}}\right)$ and $\left\{ {{\mathsf{span}}\left( {{{\bf{V}}^{[k]}}} \right)} \right\}_{k = 1}^K$ are concatenated with each other. From \eqref{Eq:IA:ex_cond}, IA precoders are computed by 
\begin{equation}\label{Eq:IA:GS}
\begin{array}{l}
 {\bf{V}}_{}^{[1]}  = d\rm{\,\,eigenvectors\,\, of}\left(\left({{\bf{H}}^{[(K-1)\,1]} } \right)^{ - 1} {\bf{H}}^{[(K-1)\,K]} \cdots \left( {{\bf{H}}^{[13]} } \right)^{ - 1} {\bf{H}}^{[12]}\left( {{\bf{H}}^{[K2]} } \right)^{ - 1} {\bf{H}}^{[K1]} \right) \\ 
 {\bf{V}}_{}^{[2]}  = \left( {{\bf{H}}^{[K2]} } \right)^{ - 1} {\bf{H}}^{[K1]} {\bf{V}}_{}^{[1]}  \\ 
  {\bf{V}}_{}^{[3]}  = \left( {{\bf{H}}^{[13]} } \right)^{ - 1} {\bf{H}}^{[12]} {\bf{V}}_{}^{[2]}  \\ 
\,\,\,\,\,\,\,\,\,\,\,\, \vdots  \\  
{\bf{V}}_{}^{[K]}  = \left( {{\bf{H}}^{[(K-2)\,K]} } \right)^{ - 1} {\bf{H}}^{[(K-2)\,(K-1)]} {\bf{V}}_{}^{[K-1]}  \\ 
 \end{array}
\end{equation}
and then each column of the precoders is normalized to have unit norm. From the design of IA precoders in \eqref{Eq:IA:GS}, total $(K-1)d$ dimensional interferers are shrunk into the $(K-2)d$ dimensional subspace at each receiver. Since the desired signals occupy $d$ dimensions of the $M$ dimensional receive space, the number of antennas should satisfy at least $M=(K-2)d+d$ for the proposed IA design. Given these antenna configurations, we can achieve $d$ DoF for each user under the design of ZF receive filter. 

\subsection{Feedback Structure}
 In the existing IA literature, the design of feedback topology is not explicitly addressed. Existing works  \cite{Thukral:LimitedSISO:2009},\cite{Varanasi:LimitedMIMO:2009} commonly assume that each receiver feeds back the estimated CSI to all transmitters. This corresponds to a \emph{full-feedback topology} as illustrated in Fig.~\ref{Fig:Topology:FullFeedback}. We consider the full-feedback topology as the conventional feedback approach for achieving IA and measure its efficiency as the CSI overhead
\begin{equation}\label{Eq:IA:metric}
N = \sum_{m, k\in \{1, 2, \cdots, K\}} \l(N_{\textsf{TR}}^{[mk]}+N_{\textsf{RT}}^{[mk]}\r)
\end{equation}
where  $N_{\textsf{TR}}^{[mk]}$ denotes the number of complex CSI coefficients sent from receiver $k$ to transmitter $m$ and  $N_{\textsf{RT}}^{[mk]}$ from transmitter $k$ to receiver $m$. { According to the feedback approach in \cite{Thukral:LimitedSISO:2009}, each receiver feeds back all interfering channels by broadcasting $(K-1) M^2$ complex coefficients to all other nodes assuming no errors. The total CSI overhead, namely the number of channel coefficients exchanged over the network, is given by
\begin{equation}\label{Prop:FF:Overhead}
N_{\textsf{FF}} = K(K-1) M^2
\end{equation}
where the overhead $N_{\textsf{FF}}$ increases as $\mathcal{O}(K^{2}M^{2})$ whereas the network throughput grows linearly with $K$. Thus the CSI overhead may outweigh the resultant throughput gain for large $K$.}
\begin{figure}[t]
\centering\includegraphics[width=11cm]{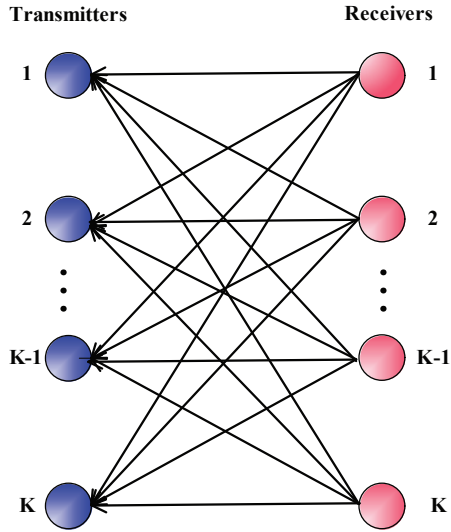}
\caption{Full-feedback  topology for achieving IA. }
\label{Fig:Topology:FullFeedback}
\end{figure} 

%%%%%%%%%%%%%%%%%%%%%%%%%%%%%%%%%%%%%%%%%%%%%%%%%%%%%
\section{CSI-Feedback Topologies}
The conventional IA technique potentially leads to unacceptable CSI-feedback overhead given the existence of many feedback links. To reduce the requirement of global CSI, the numerical methods have been proposed in the iterative algorithm which progressively  update the transmit/receive filters by using local channel knowledge at each node  \cite{Gomadam:iterativeMIMO:2008}. This iterative method  achieves the full DoF under the feasibility condition $M \ge d(K+1)/2$, but it results in a slow convergence rate that causes the huge amount of system overhead. In this section, we propose two practical CSI feedback topologies, namely  { the \emph{centralized-feedback} and \emph{CSI-exchange} topologies.}  The design of proposed feedback topologies build upon the closed-form IA solution in \eqref{Eq:IA:GS} and the efficiency is measured by the metric in \eqref{Eq:IA:metric}.  To implement the proposed feedback topologies, we make the following assumptions:
\begin{itemize}

 \item[1)] {Centralized-feedback  topology: In order to design the centralized-feedback topology, we assume that each receiver directly exchange the estimated CSI with others. This framework is feasible for the receivers who are located close together and linked with local area networks such as Wi-Fi [26], [27]. Moreover, the uplink coordinated multi-point (CoMP) system which provides the high-capacity backhaul links between base stations can be applicable for this scenario}.

 \item[2)] {CSI-exchange topology :} CSI can be exchanged in both direction between a transmitter and receiver through feedforward/feedback channels. The effect of quantization error due to the limited feedforward/feedback channels is discussed in next section.

\end{itemize}

\begin{figure}[t]
\centering\includegraphics[width=12cm]{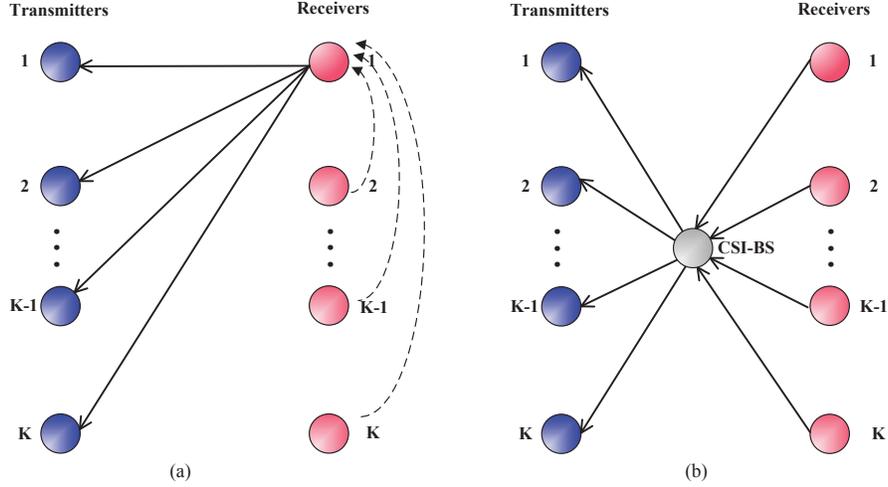}
\caption{Centralized-Feedback Topologies.}
\label{Fig:Topology:Centralized}
\end{figure}

\subsection{Centralized-Feedback Topology}

Consider $K=4$ user interference channel with $M=3d$. From \eqref{Eq:IA:GS}, the IA precoders ${\bf {V}}^{[1]}$,${\bf {V}}^{[2]}$, ${\bf {V}}^{[3]}$ and ${\bf {V}}^{[4]}$ are represented as
\begin{equation}\label{Eq:IA3:GS}
\begin{array}{l}
 {\bf{V}}_{}^{[1]}  = d\rm{\,\,eigenvectors\,\, of} \left(\left( {{\bf{H}}^{[31]} } \right)^{ - 1} {\bf{H}}^{[34]} \left( {{\bf{H}}^{[24]} } \right)^{ - 1} {\bf{H}}^{[23]}\left( {{\bf{H}}^{[13]} } \right)^{ - 1} {\bf{H}}^{[12]}\left( {{\bf{H}}^{[42]} } \right)^{ - 1} {\bf{H}}^{[41]} \right) \\ 
 {\bf{V}}_{}^{[2]}  = \left( {{\bf{H}}^{[42]} } \right)^{ - 1} {\bf{H}}^{[41]} {\bf{V}}_{}^{[1]}  \\ 
 {\bf{V}}_{}^{[3]}  = \left( {{\bf{H}}^{[13]} } \right)^{ - 1} {\bf{H}}^{[12]} {\bf{V}}_{}^{[2]}  \\ 
 {\bf{V}}_{}^{[4]}  = \left( {{\bf{H}}^{[24]} } \right)^{ - 1} {\bf{H}}^{[23]} {\bf{V}}_{}^{[3]}  \\ 
 \end{array}
\end{equation}
and normalized to unit norm at each column. { As shown in \eqref{Eq:IA3:GS}, the set of product channel matrices $\{ \left( {{\bf{H}}^{[13]} } \right)^{ - 1} {\bf{H}}^{[12]}$, $\left( {{\bf{H}}^{[24]} } \right)^{ - 1} {\bf{H}}^{[23]}$,  $\left( {{\bf{H}}^{[31]} } \right)^{ - 1} {\bf{H}}^{[34]}$, $\left( {{\bf{H}}^{[42]} } \right)^{ - 1} {\bf{H}}^{[41]} \}$  are commonly used for computing all IA precoders. By allowing CSI exchange between receivers, we propose the feedback topology where a particular receiver collects CSI from all other receivers, computes all precoders and send them to corresponding transmitters. This topology is called \emph{centralized-receiver feedback topology} as illustrated in Fig. \ref{Fig:Topology:Centralized} (a). Without loss of generality, let  receiver $1$ be the one that collects CSI form others to compute precoders. This results in two-hop feedback channels as follows: (i) the feedback channels that each of $(K-1)$ receivers sends the interfering matrix comprising $M^2$ coefficients to receiver $1$  and (ii) the feedback channels that receiver $1$  transmits a precoder of $Md$ coefficients to each of $K$ transmitters. Combining the overhead in (i) and (ii), we obtain the CSI  overhead in the centralized-receiver feedback topology as  
\begin{equation}\label{Prop:CF:Overhead}
N_{\textsf{CF}} = (K-1)M^2 + KMd.
\end{equation}}
 However,  a huge burden of computation and feedback overhead are centralized at receiver $1$ in the proposed topology. To address these issues, we propose the \emph {star feedback topology} illustrated in Fig.~\ref{Fig:Topology:Centralized} (b) and describe details in Algorithm \ref{Algorithm_star}. The star feedback topology comprises an agent, called the \emph{CSI base station} (CSI-BS)  which collects CSI from all receivers, computes all precoders using IA condition in \eqref{Eq:IA:ex_cond} and sends them back to corresponding transmitters. Similar to the centralized-receiver feedback, the CSI overhead for the star feedback topology is computed as
\begin{equation}\label{Prop:SF:Overhead}
N_{\textsf{SF}} = KM^2 + KMd.
\end{equation}

From \eqref{Prop:CF:Overhead} and \eqref{Prop:SF:Overhead}, the CSI overhead of centralized-feedback topologies is  scaled with $\mathcal{O}{(KM^2)}$.    

\begin{remark}
{ In the case that CSI sharing is valid for the transmitters, we design $\emph{centralized-transmitter}$ feedback, where the interfering channels from all receivers are fed back to the particular transmitter and then all precoders are computed and exchanged with other transmitters. While the computed precoders at receiver 1 are fed back to the corresponding transmitters in the centralized-receiver feedback topology, the centralized-transmitter feedback topology requires the feedback of interfering channel matrices from receivers to transmitters.}
\end{remark}
\begin{algorithm}[t!] \label{Algorithm_star}
\,\,
\textbf{1. Computation of ${\bf{V}}^{[1]} ,...,{\bf{V}}^{[K]}$} : {The CSI-BS collects ${\bf H}_{c}^{[k]}=\left({\bf H}^{[k \bar k]}\right)^{-1}{\bf H}^{[k \hat k]}$ from the receiver $k$, $\forall k$, where $\bar k = \textsf{mode}(k+1, K)+1$, $\hat k = \textsf{mode}(k, K)+1$ and \textsf{mod}$(n,k)$ represents the modulo operation.
}

\textbf{2. Broadcasting ${\bf{V}}^{[1]} ,...,{\bf{V}}^{[K]}$} : CSI-BS transmits  ${\bf{V}}^{[k]}$ to the corresponding transmitter $k$, $\forall k$. 
\,\,\,
\caption{Star feedback topology}
\end{algorithm}

\begin{figure}[t]
\centering\includegraphics[width=11cm]{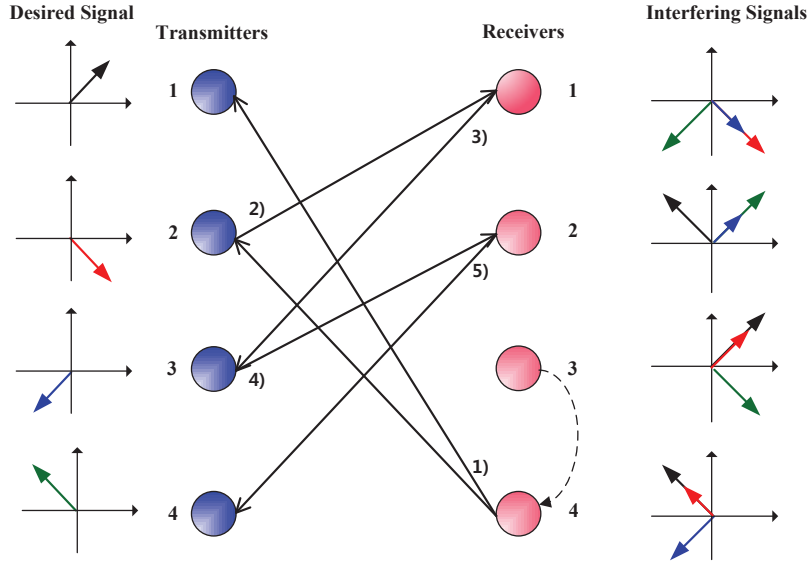}
\caption{CSI-exchange topology for achieving IA with $K$=$4$ and $M$=$3$.}
\label{Fig:Topology:CSI_exchange}
\end{figure}

\subsection{CSI-Exchange Topology $(K \ge 4)$}
In the centralized-feedback topology, ${\bf {V}}^{[1]}$ is solved by the eigenvalue problem that incorporates the channel matrices of all interfering links which causes a significant overhead for the case of many links or antennas. 
To reduce the CSI overhead for the computation of $ {\bf V}^{[1]}$, we design two interferers from transmitter $1$ and $2$ are aligned in the same subspace at receiver $(K-1)$ and $K$ as following conditions:
\begin{equation}\label{Eq:IA:Mex}
\begin{array}{l}
 \mathsf{span} ({\bf{H}}^{[(K-1)\,1]} {\bf{V}}^{[1]}) =\mathsf{span} ({\bf{H}}^{[(K-1)\,2]} {\bf{V}}^{[2]})\,\,\,\,\,{\rm{at}}\,\,{\rm{receiver}}\,(K-1) \\ 
\,\,\,\,\,\,\,\,\, \mathsf{span} ({\bf{H}}^{[K\,1]} {\bf{V}}^{[1]}) = \mathsf{span} ({\bf{H}}^{[K\,2]} {\bf{V}}^{[2]})\,\,\,\,\,\,\,\,\,\,\,\,\,\,\,{\rm{at}}\,\,{\rm{receiver}}\,K  \\
 \end{array}
\end{equation}
Substituting \eqref{Eq:IA:Mex} with last two conditions in \eqref{Eq:IA:ex_cond}, IA precoders ${\bf{V}}^{[1]},{\bf{V}}^{[2]},...,{\bf{V}}^{[K]}$ are modified as  
\begin{equation}\label{Eq:IA:M_V_1_K}
\begin{array}{l}
 {\bf{V}}_{}^{[1]}\,\,{\rm{ = }}\,\,\,d\rm{\,\,eigenvectors\,\, of}\,\left({\left({{{\bf{H}}^{[(K - 1)\,1]}}} \right)^{ - 1}}{{\bf{H}}^{[(K - 1)\,2]}}{\left( {{{\bf{H}}^{[K2]}}} \right)^{ - 1}}{{\bf{H}}^{[K1]}}\right) \\ 
 {\bf{V}}_{}^{[2]} = {\left( {{{\bf{H}}^{[K2]}}} \right)^{ - 1}}{{\bf{H}}^{[K1]}}{\bf{V}}_{}^{[1]} \\ 
 {\bf{V}}_{}^{[3]} = {\left( {{{\bf{H}}^{[13]}}} \right)^{ - 1}}{{\bf{H}}^{[12\,]}}{\bf{V}}_{}^{[2]} \\ 
\,\,\,\,\,\,\,\,\,\,\,\, \vdots  \\ 
 {\bf{V}}_{}^{[K]} = {\left( {{{\bf{H}}^{[(K - 2)\,K]}}} \right)^{ - 1}}{{\bf{H}}^{[(K - 2)\,(K - 1)]}}{\bf{V}}_{}^{[K - 1]} . \\ 
 \end{array}
\end{equation}
From \eqref{Eq:IA:M_V_1_K}, the $k$-th precoder ${\bf{V}}_{}^{[k]}$ is sequentially determined by the product of pre-determined ${\bf{V}}_{}^{[k-1]}$ and the estimated channel matrix ${{{\left( {{{\bf{H}}^{[(k-2)\,k]}}} \right)}^{ - 1}}{{\bf{H}}^{[(k-2)\,(k-1)]}}}$ at receiver $(k-2)$ for $k\ge 3$. These properties motivate the design of sequential CSI-exchange topology in Algorithm \ref{Exchange_1}, which only exchanges precoding matrices between transmitters and receivers after the computation of ${\bf{V}}^{[1]}$ and ${\bf{V}}^{[2]}$. { Fig.~\ref{Fig:Topology:CSI_exchange} illustrates the CSI-exchange topology for $K=4$ user interference channel and its procedure, ${R_4}\mathop  \to \limits^{{{\bf{V}}^{[1]}}}_{1)} {T_1}$, ${R_4}\mathop  \to \limits^{{{\bf{V}}^{[2]}}}_{1)} {T_2}\mathop  \to \limits^{{{\bf{V}}^{[2]}}}_{2)}  {R_1}\mathop  \to \limits^{{{\bf{V}}^{[3]}}}_{3)}  {T_3}\mathop  \to \limits^{{{\bf{V}}^{[3]}}}_{4)}  {R_2}\mathop  \to \limits^{{{\bf{V}}^{[4]}}}_{5)}  {T_4}$, }where $T_m$ and $R_n$ represent transmitter $m$ and receiver $n$, respectively.

\begin{algorithm}[t]\label{Exchange_1}
\,\,

\textbf{1. Computation of ${\bf{V}}^{[1]}$ and ${\bf{V}}^{[2]}$} 

 {Receiver $(K-1)$ forwards the matrix ${\bf{H}}_{e}^{[K-1]}$ to receiver $K$. Then, receiver $K$ computes ${\bf {V}}^{[1]}$ and  ${\bf {V}}^{[2]}$  using \eqref{Eq:IA:M_V_1_K} and feeds them back to the corresponding transmitter $1$ and $2$.\\}

\textbf{2. Exchange of precoders ${\bf{V}}^{[2]} ,...,{\bf{V}}^{[K]}$} \\
\, \For{$k$=$2$:$(K-1)$}{
Transmitter $k$ forwards ${\bf{V}}^{[k]}$ to receiver $(k-1)$. Then, receiver $(k-1)$ calculates ${\bf{V}}^{[k+1]}$ and feeds back to transmitter $(k+1)$.\\
}
\,\,\,

\,\,\,
\caption{CSI-exchange topology}
\end{algorithm}
\,\,

Let denote  ${\bf{H}}_{e}^{[K-1]} = {\left({{{\bf{H}}^{[(K - 1)\,1]}}} \right)^{ - 1}}{{\bf{H}}^{[(K - 1)\,2]}}$ and ${\bf{H}}_{e}^{[K]} ={\left( {{{\bf{H}}^{[K2]}}} \right)^{ - 1}}{{\bf{H}}^{[K1]}}$, respectively.
In Algorithm \ref{Exchange_1},  receiver $(K-1)$ transmits CSI of the product channel matrix ${\bf{H}}_{e}^{[K-1]}$ to receiver $K$, which comprises $M^2$ complex-valued coefficients. Using ${\bf{H}}_{e}^{[K-1]}$ and ${\bf{H}}_{e}^{[K]}$, receiver $K$ computes the $Md$ complex-valued  precoders ${\bf{V}}^{[1]} $ and ${\bf{V}}^{[2]}$, and feeds them back to the corresponding transmitter $1$ and $2$. Then each precoder is determined by iterative exchange of precoders between transmitters and interfered receivers. In each round of exchange, the number of nonzero coefficients of feedforward/feedback becomes $2Md$. Therefore, total CSI overhead in the CSI-exchange topology requires 
\begin{align}
N_{\textsf{EX}} &= M^2+ 2Md+2(K-2)Md \nonumber\\
&=M^2+2(K-1)Md. \label{Prop:MEX:Overhead}
\end{align}

 From \eqref{Prop:MEX:Overhead}, the proposed topology provides much less CSI overhead for achieving IA, namely on the order of $KM$, whereas the conventional feedback approach requires the feedback overhead of $K^2M^2$ order.  Comparing \eqref{Prop:MEX:Overhead} with \eqref{Prop:CF:Overhead} and \eqref{Prop:SF:Overhead},  the product channel matrices for computing ${\bf{V}}^{[1]}$ in the CSI-exchange topology requires constant $M^2$ overhead in any $K$ user cases while that of centralized-feedback topologies increase with $KM^2$.

\subsection {Comparison of Centralized-Feedback and CSI-Exchange Topology}
While the CSI-exchange topology degrades the amount of feedback overhead compared with the full-feedback topology, it  incur $2(K-1)$ iterations caused by the procedure of multiple-hop feedforward/feedback between the transmitter and the receiver. As the number of iterations is increased, the full DoF in $K$-user interference channel can not be achievable since it causes the time delay of transmission that results in significant interference misalignment for fast fading. However, the centralized-feedback topology is robust against channel variations  as it requires only two time slots for the computation of IA precoders in any number of user $K$.  Compared with the CSI-exchange topology, CSI-BS that connects all pairs of transmitter-receiver should be implemented as the additional costs in the centralized-feedback topology. In addition, the feedback overhead is increased with $\mathcal O \left(KM^2\right)$ which is larger than $\mathcal O \left(KM\right)$ in the CSI-exchange topology.
  Fig. \ref{Fig:Topologyr:Overhead} compares the CSI overhead of proposed feedback topologies for $d=1$ scenarios. We figure out that the full-feedback topology provides dramatic CSI overhead compared with the proposed feedback topologies, while centralized-feedback topologies show slightly larger overhead than that of CSI-exchange topology. 
{
\begin{figure}[t]
\centering\includegraphics[width=12cm]{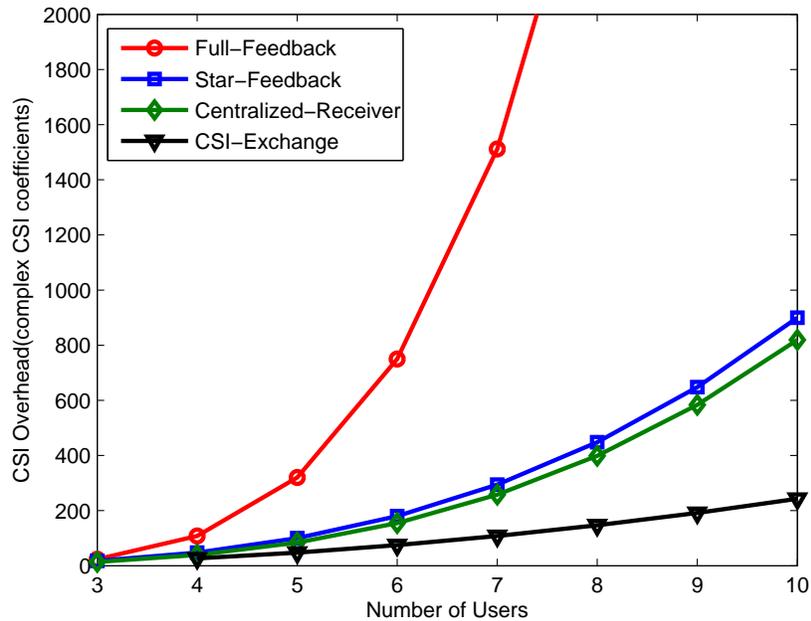}
\caption{{ Comparison of CSI-feedback overhead in the proposed topologies.}}
\label{Fig:Topologyr:Overhead}
\end{figure} 
}

\section{Effect of CSI-Feedback Quantization} 
The proposed feedback topologies are designed under the assumption of perfect CSI exchange in the preceding section. However, in practical communication systems, CSI is quantized at each receiver and sent back to the corresponding transmitter through the finite-rate feedback constraints which causes the performance degradation due to the residual interference at receive sides. In this section, we analyze the throughput loss in the proposed feedback topologies due to the limited feedback channels \cite{Jindal:BD:2008}. The RVQ is used for CSI quantization and a single data transmission $d=1$ is considered  for analytical simplicity. 

\subsection{Throughput Loss Analysis}

Prior to deriving the throughput loss, we quantify the quantization error with RVQ using the distortion measure. Let denote a ${M \times 1}$ beamformer ${\bf{v}}^{[k]}$ at transmitter $k$, satisfying ${\left\| {\bf{v}}^{[k]} \right\|^2}=1$ and the quantization codebook  $\mathcal{W}$ known to both transmitters and receivers. Given $B_k$ feedback-bits, the codebook $\mathcal{W}$ consists of $2^{B_k}$ independently selected random vectors from the isotropic distribution on the $M$ dimensional complex unit sphere, where $\mathcal{W} = \left\{ {{{{\bf{\hat v}}}}_1,...,{{{\bf{\hat v}}}_{{2^{B_k}}}}} \right\}$. The quantized beamformer ${\bf {\hat v}}^{[k]}$ is selected by the minimal chordal distance metric:
\begin{equation} \label{Eq:IA_Chordal}
 {\bf{\hat v}}^{[k]} = \mathop {\arg \min }\limits_{{{{\bf{\hat v}}}_i} \in \mathcal{W}} {d^2}\left( {{\bf{v}}^{[k]},{{{\bf{\hat v}}}_i}} \right),
 \end{equation} 
where $d\left( {{\bf{v}}^{[k]},{{{\bf{\hat v}}}_i}} \right) = \sin \theta_k= \sqrt {1 - \left| {{{\bf{v}}^{[k]\dag} }{{{\bf{\hat v}}}_i}} \right|^2}$ and $\theta_k $ denotes the principle angle between ${\bf{v}}^{[k]}$ and ${\bf{\hat v}}_i$.
{Using the quantized beamformer ${\bf{\hat v}}^{[k]}$, we model the ${\bf{v}}^{[k]}$ as
\begin{align}
{\bf{\hat v}}^{[k]} &= (\cos \theta_k) {\bf{v}}^{[k]}\, + (\sin \theta_k) \Delta {\bf{v}}^{[k]} \nonumber \\
 &= \sqrt {1-\sigma_k} \, {\bf{v}}^{[k]}\, + \sqrt \sigma_k \Delta {\bf{v}}^{[k]}, \label{Eq:IA:h_model}
\end{align}            		         
where $\Delta {\bf{v}}^{[k]}$ represents the quantization error of ${\bf{v}}^{[k]}$ with unit norm and  $\sigma_k { \buildrel \Delta \over =} \sin^2\theta_k$. From \cite{Dai:Bound:2005}, the upper-bound of quantization distortion is given by $\mathbb E\left[ \sigma_k \right]\le \bar \Gamma(M) \cdot 2^{-{{B_k}\over {M-1}}}$, where $\bar \Gamma(M)={{\Gamma({1 \over {M-1}})}\over{ M-1}}$ and $\Gamma (x)$ represents the gamma function of $x$.}

 Let denote ${\mathbf{\hat{v}}}^{[j]}$ and ${\mathbf{\hat{r}}}^{[k]}$ as the transmit beamformer and receive filter calculated in the presence of CSI quantization errors. Then, we write the residual interference power at  receiver $k$ as
 \begin{equation} 
{{\hat I}^{[k]}} = \sum_{k \ne j} Pd_{kj}^{ - \alpha }{\left| {{{{\bf{\hat r}}}^{[k]\dag }}{{\bf{H}}^{[kj]}}{{{\bf{\hat v}}}^{[j]}}} \right|^2}
\end{equation}
and the sum throughput as 
\begin{equation}\label{Eq:IA:R_sum}
R_{\textsf{limited}} = \sum\limits_{k = 1}^K {{{\log }_2}} \left( {1 + \frac{{{P}{d_{kk}^{-\alpha} }{{\left| {{{{\bf{\hat r}}}^{[k]\dag }}{{\bf{H}}^{[kk]}}{{{\bf{\hat v}}}^{[k]}}} \right|}^2}}}{{{\hat I}^{[k]}   + 1}}} \right).
\end{equation}
In this subsection, we derive the throughput loss $\Delta R_{sum}$, which represents the difference between the sum throughput by perfect CSIT-based IA and limited feedback-based IA: $\Delta R_{\textsf{sum}} = {\mathbb E_{\bf{H},{\mathcal{W}}}}[R_{\textsf{perfect}}-R_{\textsf{limited}}]$.  Then, the throughput loss is upper bounded as 
\begin{equation}\label{Eq:IA:Delta_R_SUM}
\begin{array}{l}
 \Delta {R_{{\rm{sum}}}} = \,\,{\mathbb E_{\bf{H}}}\left[ {\sum\limits_{k = 1}^K {{{\log }_2}\left( {1 + Pd_{kk}^{ - \alpha }{{\left| {{{\bf{r}}^{[k]\dag }}{{\bf{H}}^{[kk]}}{{\bf{v}}^{[k]}}} \right|}^2}} \right)} } \right] - {\mathbb E_{{\bf{H}},{\mathop{\mathcal W}\nolimits} }}\left[ {\sum\limits_{k = 1}^K {{{\log }_2}\left( {1 + \frac{{Pd_{kk}^{ - \alpha }{{\left| {{{{\bf{\hat r}}}^{[k]\dag }}{{\bf{H}}^{[kk]}}{{{\bf{\hat v}}}^{[k]}}} \right|}^2}}}{{1 + {{\hat I}^{[k]}}}}} \right)} } \right] \\ 
 \,\,\, \,\, \,\,\,\,\,\,\,\,\,\,\,\,\,\,\, = \,\,{\mathbb E_{\bf{H}}}\left[ {\sum\limits_{k = 1}^K {{{\log }_2}\left( {1 + Pd_{kk}^{ - \alpha }{{\left| {{{\bf{r}}^{[k]\dag }}{{\bf{H}}^{[kk]}}{{\bf{v}}^{[k]}}} \right|}^2}} \right)} } \right] \\ 
 \,\,\, \,\,\, \,\,\, \,\,\,\,\,\,\,\,\,\,\,\,\,\,\, - {\mathbb E_{{\bf{H}},{\mathop{\mathcal W}\nolimits} }}\left[ {\sum\limits_{k = 1}^K {{{\log }_2}\left( {1 + {{\hat I}^{[k]}} + Pd_{kk}^{ - \alpha }{{\left| {{{{\bf{\hat r}}}^{[k]\dag }}{{\bf{H}}^{[kk]}}{{{\bf{\hat v}}}^{[k]}}} \right|}^2}} \right)} } \right] + {\mathbb E_{{\bf{H}},{\mathop{\mathcal W}\nolimits} }}\left[ {\sum\limits_{k = 1}^K {{{\log }_2}\left( {1 + {{\hat I}^{[k]}}} \right)} } \right] \\ 
 \,\,\,  \,\,\,\,\,\,\,\,\,\,\,\,\,\,\,\,\,\,\mathop { \le \,}\limits^{} {\mathbb E_{\bf{H}}}\left[ {\sum\limits_{k = 1}^K {{{\log }_2}\left( {1 + Pd_{kk}^{ - \alpha }{{\left| {{{\bf{r}}^{[k]\dag }}{{\bf{H}}^{[kk]}}{{\bf{v}}^{[k]}}} \right|}^2}} \right)} } \right]  \\ 
 \,\,\,\,\,\,\,\,\,\,\,\,\,\,\,\,\,\,\,\,\,\,\,\,-{\mathbb E_{{\bf{H}},{\mathop{\mathcal W}\nolimits} }}\left[ {\sum\limits_{k = 1}^K {{{\log }_2}\left( {1 + Pd_{kk}^{ - \alpha }{{\left| {{{{\bf{\hat r}}}^{[k]\dag }}{{\bf{H}}^{[kk]}}{{{\bf{\hat v}}}^{[k]}}} \right|}^2}} \right)} } \right] + {\mathbb E_{{\bf{H}},{\mathop{\mathcal W}\nolimits} }}\left[ {\sum\limits_{k = 1}^K {{{\log }_2}\left( {1 + {{\hat I}^{[k]}}} \right)} } \right] \\ 
 \,\,\, \,\,\, \,\,\,\,\,\,\,\,\,\,\,\,\,\,\,\mathop  = \limits^{(a)} {\mathbb E_{{\bf{H}},{\mathop{\mathcal W}\nolimits} }}\left[ {\sum\limits_{k = 1}^K {{{\log }_2}\left( {1 + {{\hat I}^{[k]}}} \right)} } \right] \\ 
 \,\,\, \,\,\,\,\,\,\,\,\,\,\,\,\,\,\,\,\,\,\mathop { \le \,}\limits^{(b)} {\mathbb E_{{\bf{H}},{\mathop{\mathcal W}\nolimits} }}\left[ {K \cdot {{\log }_2}\left( {1 + \frac{1}{K}\sum\nolimits_{k = 1}^K {{{\hat I}^{[k]}}} } \right)} \right] \\ 
 \end{array}
\end{equation} 
where (a) follows the fact that ${\bf{v}}^{[k]}$, ${\bf{r}}^{[k]}$, ${\bf{\hat v}}^{[k]}$ and ${\bf{\hat r}}^{[k]}$ are independently distributed in $\mathbb{C}^{M \times 1}$  and (b) uses the characteristic of concave function, $\log(x)$. 
Applying Jensen's inequality to the upper-bound in \eqref{Eq:IA:Delta_R_SUM}, the throughput loss is upper-bounded by
\begin{equation}\label{Eq:IA:Delta_bound}
\Delta R_{\textsf{sum}}^{}\mathop  \le \limits^{} K\cdot{\log _2}\left(1+{\frac{1}{K}{\mathbb E_{{\bf{H}},{\mathcal{W}}}}\left[ {\sum\nolimits_{k = 1}^K {{{\hat I}^{[k]}}} } \right]} \right).
\end{equation}
This bound  explains that the throughput loss is logarithmically increased with the sum of residual interference. To minimize the throughput loss due to the quantization error, we analyze the residual interference at each receiver and regulate it by utilizing a variable feedback-bits allocation schemes in following sections.
\subsection{Residual Interference Relative to Quantization Error}
\subsubsection{Centralized-Feedback  Topology}
{  In the centralized-receiver (star) feedback topology, we assume that all receivers are connected to receiver 1 (CSI-BS) with high-capacity backhaul links which allows receiver 1 (CSI-BS) to acquire full knowledge of CSI estimated at each receiver. Then, receiver 1 (CSI-BS) can compute ${\bf{v}}^{[1]},{\bf{v}}^{[2]},...,{\bf{v}}^{[K]}$ using~\eref{Eq:IA:GS} and forwards them to the corresponding transmitters for achieving IA. Given that the feedback of ${\bf{v}}^{[k]}$ has $B_k$ bits, the expected residual interference at each receiver can be upper bounded as shown below.

\begin{proposition}\label{Prop:STAR:E_SUM_I}
In the centralized-feedback topology, the expected residual interference at each receiver can be upper bounded as
\begin{equation}\label{Eq:STAR:E_1_K}
{\mathbb E_\mathcal{{\bf {H}},  W}}\left[ {{{\hat I}^{[k]}}} \right] \le {\bar \Gamma(M)}\cdot \left( Pd_{k \hat k}^{ - \alpha } \cdot M^2 \cdot 2^{- {B_{\hat k}\over {M-1}}} + Pd_{k \bar k}^{ - \alpha }\cdot M^2 \cdot 2^{- {B_{\bar k}\over {M-1}}}\right), \,\, \forall \, k
\end{equation}
 given the number of feedback bits $\left\{ {{B_{k}}} \right\}_{k = 1}^K$.
\end{proposition}
\begin{proof}
See Appendix~\ref{1}.
\end{proof}
}
In Proposition \ref{Prop:STAR:E_SUM_I}, the residual interference at receiver $k$ is generated by the misalignment between interference from transmitter $\hat k$ and $\bar k$. The upper-bound of expected residual interference varies from the number of feedback-bits at each receiver and the path-loss between the pairs of transmitter-receiver.

\subsubsection {CSI-Exchange Topology}
 Under the finite-rate feedforward/feedback channels between transmitters and receivers, we analyze the residual interference in  CSI-exchange topology that consists of two types of CSI exchange links: (i) exchange of the channel matrix between receiver $(K-1)$ and $K$ and (ii) sequential exchange of quantized beamformer between transmitters and receivers through feedforward/feedback links. For tractability, we assume that receiver $(K-1)$ and $K$ are located in local areas and linked with high-capacity Wi-Fi links \cite{Shin:UserCoop:2010, Draft:80216:2010}. This connectivity is feasible in the integrated heterogeneous network (e.g., Wi-Fi / cellular) which provides user cooperation in short-range area so that the perfect CSI sharing is allowed to both receivers \cite{Nusairat:Hetero:2007}.  Then, receiver $K$ computes both ${{\bf{v}}^{[1]}}$ and ${{\bf{v}}^{[2]}}$, satisfying {${\bf{v}}_{r}^{[K]}={\bf{H}}^{[K1]}{{\bf{v}}^{[1]}}={\bf{H}}^{[K2]}{{\bf{v}}^{[2]}}$ and ${\bf{v}}_{r}^{[K-1]}={\bf{H}}^{[(K-1)\,1]}{{\bf{v}}^{[1]}}={\bf{H}}^{[(K-1)\,2]}{{\bf{v}}^{[2]}}$}, and feeds back the quantized beamformer ${\bf{\hat v}}_{}^{[1]}$ and ${\bf{\hat v}}_{}^{[2]}$ to transmitter $1$ and $2$, chosen according to \eqref{Eq:IA_Chordal}. Next, transmitter $2$ forwards ${\bf{\hat v}}_{}^{[2]}$ to receiver $1$.  For ${\bf{\hat v}}_{}^{[k]}$, $k=3,...,K$, receiver $(k-2)$  sequentially designs ${\bf{ v}}_{}^{[k]}$ to be aligned with {${\bf v}_r^{[k-2]}={{{\bf{H}}^{[(k-2)\, (k-1)]}}} {\bf{\hat v}}_{}^{[k-1]}$} where $span\left( {\bf{ v}}_{}^{[k]}\right)=span\left( {{\left( {{{\bf{H}}^{[(k-2)\,k]}}} \right)}^{ - 1}{\bf v}_r^{[k-2]} }\right)$ and feeds back the quantized ${{\bf{\hat v}}^{[k]}}$ to transmitter $k$.  In following Proposition, the upper-bound of residual interference averaged over all random choices of codebooks and channels are derived in the CSI-Exchange topology.
 {
\begin{proposition}\label{Prop:H:E_SUM_I}
In  the CSI-exchange topology, the expected residual interference at each receiver can be upper bounded as 
\begin{equation}\label{Eq:IA:E_I_k}
\left\{ \begin{array}{l}
 {\mathbb E_{ {\bf H}, \mathcal{W}}}\left[ {{{\hat I}^{[k]}}} \right]\,\,\,\,\,\,\, \le {\bar \Gamma(M)}\cdot \left( Pd_{k\,(k+2)}^{ - \alpha }\cdot M^2 \cdot 2^{- {B_{k+2} \over {M-1}}}\right ),\,\,\,\,\,k = 1,...,(K - 2) \\ 
 {\mathbb E_{ {\bf H}, \mathcal{W}}}\left[{{{\hat I}^{[K - 1]}}} \right] \le {\bar \Gamma(M)}\cdot \left(Pd_{(K-1) \,1}^{ - \alpha } \cdot M^2 \cdot 2^{- {B_{1}\over {M-1}}} + Pd_{(K-1)\, 2}^{ - \alpha }\cdot M^2 \cdot 2^{- {B_{2}\over {M-1}}} \right) \\  
  {\mathbb E_{ {\bf H}, \mathcal{W}}}\left[{{{\hat I}^{[K]}}} \right] \,\,\,\,\, \le {\bar \Gamma(M)}\cdot \left(Pd_{K 1}^{ - \alpha } \cdot M^2 \cdot 2^{- {B_{1}\over {M-1}}} + Pd_{K 2}^{ - \alpha }\cdot M^2 \cdot 2^{- {B_{2}\over {M-1}}}\right) \\ 
  \end{array} \right.
\end{equation} 
given the number of feedback bits $\left\{ {{B_{k}}} \right\}_{k = 1}^K$.
\end{proposition}
\begin{proof}
See Appendix~\ref{2}.
\end{proof}
}

In Proposition \ref{Prop:H:E_SUM_I}, the expected residual interference in the CSI-exchange topology is characterized as a function of the feedback-bits at each receiver  and path-loss between the pairs of transmitter-receiver. Since ${\bf { v}}^{[1]}$ and ${\bf { v}}^{[2]}$  are designed based on the IA condition \eqref{Eq:IA:Mex}, both quantized ${\bf {\hat v}}^{[1]}$ and ${\bf {\hat v}}^{[2]}$ affect the residual interference at receiver $(K-1)$ and $K$ in \eqref{Eq:IA:E_I_k}. However, other receiver $k$ sequentially designs ${\bf {\hat v}}^{[k+2]}$ based on the pre-determined ${\bf {\hat v}}^{[k+1]}$ so that the interference at receiver $k$ is only affected by the quantization error of  ${\bf {\hat v}}^{[k+2]}$, respectively. Fig. \ref{Fig:Residual:R1} depicts the residual interference at receiver $1$ and $K$ in the CSI-exchange topology, as an example. 

\begin{figure}[t]
\centering\includegraphics[width=16cm]{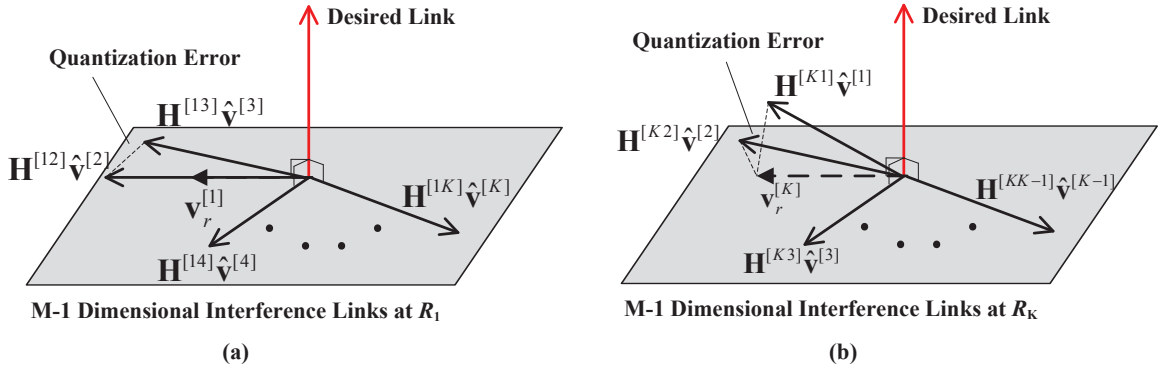}
\caption{Quantization errors at receiver $1$ and $K$ in the CSI-Exchange topology.}
\label{Fig:Residual:R1} 
\end{figure}

{
\subsection {Effect of Imperfect Local CSI Exchange}
The effect of imperfect CSI exchange between receivers is discussed in this section. In particular, the required number of CSI bits for such an exchange is derived for both the centralized-feedback and CSI-exchange topologies as follows. 

\subsubsection{Centralized-Feedback Topology}
 In limited feedback channel between receivers and CSI-BS, the receiver $k$  quantizes channel coefficients of ${\bf H}^{[k]}_{c}$ that consists of $M^2$ compelx values and sends them to the receiver 1 (CSI-BS). Using the result in \cite{Thukral:LimitedSISO:2009} that shows each of channel matrices requires $(M^2-1)\log_2P$ bits quantization to obtain the full DoF, we obtain the total number of CSI bits for the local CSI exchange as $B_M=K\cdot \left(M^2-1\right) \cdot \log_2P$ bits in the centralized-feedback topology.

\subsubsection{CSI-Exchange Topology} Under the assumption of the finite-rate feedback channel  between receiver $(K-1)$ and $K$, we analyze the quantization error of channel matrix and derive the required CSI bits for achieving the full DoF in the CSI-exchange topology. Let assume that the receiver $(K-1)$ quantizes ${\bf H}^{[K-1]}_{e}$ with $B_M$ bits random codebooks and feeds back to the receiver $K$. Using ${\bf{h}}_{e}^{[K-1]}: = {{\mathsf {vec}\left( {\bf{H}}_{e}^{[K-1]} \right)}\over {{\left\| {\bf{H}}_{e}^{[K-1]} \right\|_F}}}$ and  $M^2$ dimensional random vector codebooks $\mathcal {W}=\{ {{\bf{\hat h}}_1}, ... ,   {{\bf{\hat h}}_{2^{B_M}}}\}$, the quantized ${\bf \hat H}^{[K-1]}_{e}$ is modeled as 
 \begin{align}
{\bf{\hat H}}_{e}^{[K-1]} &=  \sqrt {1-\sigma_M} \, {\bf{H}}_{e}^{[K-1]}\, + \sqrt {\sigma_M} \Delta {\bf{H}}_{e}^{[K-1]} \label{Eq:IA:H_model}
\end{align}            		         
where $\mathsf{vec}\left( \bf H \right)$ denotes the vectorization of a matrix $\bf H$, $\Delta {\bf{H}}_{e}^{[K-1]}$ represents the quantization error with unit norm and  ${\mathbb{E}}_{}\left[\sigma_{M} \right]= \bar \Gamma {(M^2)}\cdot 2^{-{B_M \over {M^2-1}}}$.  Since $ {\bf v} ^ {[1]}$ is determined by the eigenvector of  $ {\bf{H}}_{e}^{[K-1]} {\bf{ H}}_{e}^{[K]}$ in \eqref{Eq:IA:M_V_1_K},  the quantized CSI of  ${\bf{ H}}_{e}^{[K-1]}$ causes an inaccurate computation of ${\bf  v}^{[1]}$  in limited channel feedback between receivers. The required $B_M$ bits are derived in the following lemma.
\begin{lemma}\label{Lemma:H:Q_V1}
Computation error of ${\bf v}^{[1]}$ due to the imperfect CSI exchange between receivers is upper-bounded by
\begin{equation}
{\mathbb E}\left[ {\left\| {\Delta {\bf{\bar v}}^{[1]} } \right\|^2 } \right] \le 2^{-{B_M \over {M^2-1}}}{\sum\limits_{k = 1,k \ne m}^M  {\frac{{ \left\|{{{\bf{H}}_{e}^{[K]}}}\right\|^2}}{{ \left | ({\lambda _m} - {\lambda _k}) \right |^2 }}}}.
\end{equation}
where ${{\bf{v}}^{[1]}}={{\bf{v}}_{m}}$ and ${{\bf{v}}_m}$ and $ \lambda_m$ are the $m$-th eigenvector and the corresponding eigenvalue of ${{\bf{H}}_{e}^{[K-1]}}{{\bf{H}}_{e}^{[K]}}$, respectively.
Moreover,  $B_M=(M^2-1)\cdot \log_2 P$ bits are required to achieve the full DoF in the feedback channel between receiver $(K-1)$ and $K$.
\end{lemma}
\begin{proof}
See Appendix~\ref{3}.
\end{proof}
Since the receivers are located close together and CSI is locally exchanged via high-capacity channel such as Wi-Fi, the required number of  CSI bits can be provided for both the centralized-feedback and CSI-exchange topologies. Then the limited number of feedback bits from receivers to transmitters becomes a dominant factor to degrade the system performance in the proposed feedback topologies.}

\section{Feedback-Bit Allocation Strategies} 
 In cooperative base-station systems, adaptive feedback bit partitioning between desired and interfering channel at each receiver has been proposed to minimize the mean loss of sum throughput in \cite{Zhang:Adaptive:2009, Heath:Adaptive_Bit:2009}. This motivates us to design the dynamic feedback-bit allocation strategy that adaptively distributes the number of feedback-bits to each pair of links for minimizing the throughput loss. To implement the feedback-bits allocation scheme, we consider the centralized bit controller which gathers channel gains from all receivers and computes the number of feedback-bits for each receiver.  Fig. \ref{Fig:Controller:Bit} depicts the structure of dynamic feedback-bit allocation in $K$-user interference channels under the constraints of total $B_T$ feedback-bits. 

\begin{figure}[t]
\centering\includegraphics[width=12cm]{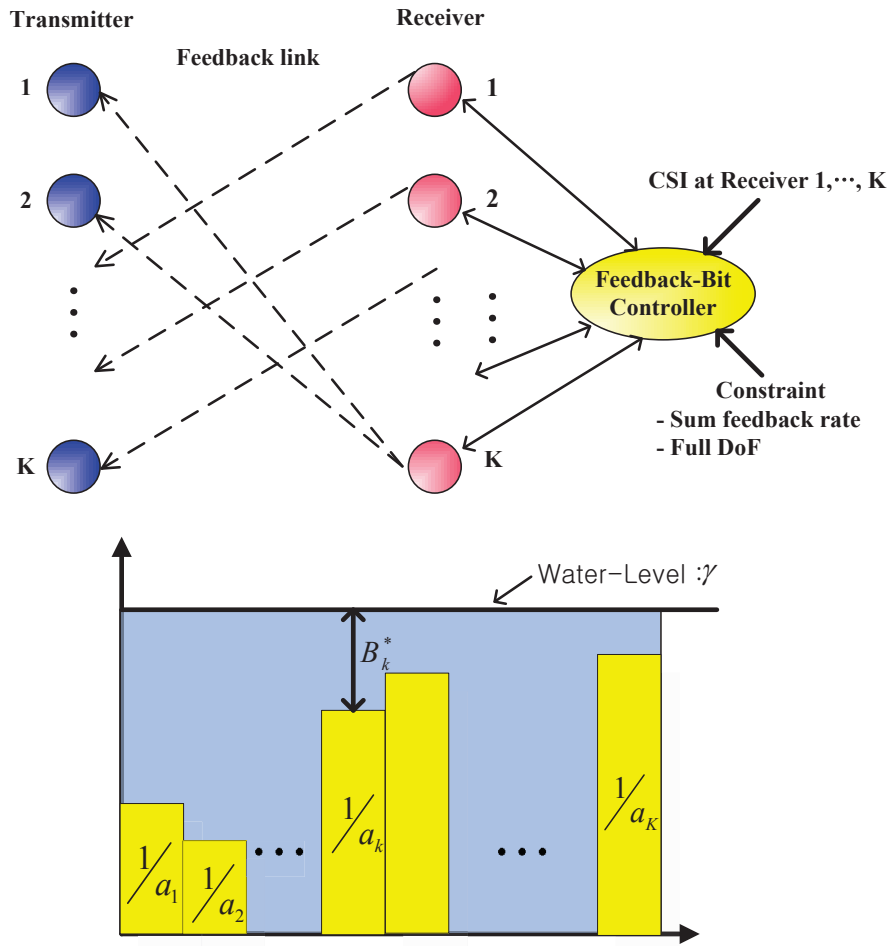}
\caption{Dynamic feedback-bits allocation scheme and the optimal bits-allocation solution. }
\label{Fig:Controller:Bit} 
\end{figure} 

\subsection{Dynamic Feedback-Bits Allocation for Minimizing Throughput Loss}
The throughput loss is characterized by the sum of residual interference in~\eref{Eq:IA:Delta_bound}. Therefore, we formulate the solution of dynamic feedback-bit allocation as following optimization problem : 
\begin{equation}\label{Eq:IA:OF}
\begin{array}{l}
 \min \limits_{\bf B} \,\,\sum\limits_{k = 1}^K {\mathbb E_{\bf H, \mathcal{W}}\left[{{{{\hat I}}}^{[k]}}\right]}  \\ 
 {\rm{s}}{\rm{.t}}{\rm{.}}\,\,\,\sum\limits_{k = 1}^K {{B_k} \le {B_{T}}}  \\ 
 \end{array}
\end{equation}
where ${\bf{B}}= \left\{ {{B_1,...,B_K}} \right\}$ are the non-negative integers. From the results in Proposition \ref{Prop:STAR:E_SUM_I} and \ref{Prop:H:E_SUM_I},  \eref{Eq:IA:OF} can be transformed to convex optimization problem :
\begin{equation}\label{Eq:IA:Convex}
\begin{array}{l}
 \min \limits_{\bf B} \,\sum\limits_{k = 1}^K {{a_k}{2^{ - \frac{{B_k^{}}}{{M - 1}}}}} \, \\ 
 {\rm{s}}{\rm{.t}}{\rm{.}}\,\,\,\sum\limits_{k = 1}^K {{B_k} \le {B_{T}}}.  \\ 
 \end{array}
\end{equation}
Here, we define $ \left\{ {{a_k}} \right\}_{k=1}^K$ in the centralized-feedback topology as
\begin{equation}\label{Eq:IA:a_k_s}
{a_k} = \overline \Gamma  \left( M \right) \cdot \left( {Pd_{\dot k k}^{ - \alpha }\cdot M^2 + Pd_{\ddot k { k}^{}}^{ - \alpha }\cdot M^2} \right)  \,\,\,\,\,\,\,\,\,\,    k=1,...,K \\ 
\end{equation}
where $ \dot k =\textsf{mod}(K+(k-3), K)+1 $ and  $\ddot {k}=\textsf{mod}(K+(k-2),K)+1$, respectively. Moreover,  $ \left\{ {{a_k}} \right\}_{k=1}^K$ in the CSI-Exchange topology is defined  as  
\begin{equation}\label{Eq:IA:a_k_m}
 \left\{ \begin{array}{l}
 {a_1} = {\bar \Gamma(M)}\cdot  \left( {Pd_{K1}^{ - \alpha }\cdot M^2 + Pd_{(K-1)\, 1}^{ - \alpha }\cdot M^2} \right) \\ 
 {a_2} = {\bar \Gamma(M)}\cdot  \left( {Pd_{(K-1)\,2}^{ - \alpha }\cdot M^2 + Pd_{K2}^{ - \alpha } \cdot M^2} \right) \\ 
 {a_k} = {\bar \Gamma(M)}\cdot  Pd_{(k - 2)\,k}^{ - \alpha }\cdot M^2, \,\,\,\,\,\,\,\,\,\,\,\,\, k  = 3,...,K. \\ 
 \end{array} \right.
\end{equation}

In order to solve the constrained optimization problem in \eref{Eq:IA:Convex}, we formulate the Lagrangian and take derivative with respect to $B_{k}$. Then, we have
\begin{equation}\label{Eq:IA:L}
L = \sum\limits_{k \in {\mathop{\rm U}\nolimits} }^{} {{a_{k}}{2^{ - \frac{{B_{k}^{}}}{{{M-1}}}}}}  + \nu \left( {\sum\limits_{k \in {\mathop{\rm U}\nolimits} }^{} {{B_{k}} - {B_{T}}} } \right)
\end{equation}
and
\begin{equation}\label{Eq:IA:L_0}
\frac{{\partial L}}{{\partial {B_{k}}}} =  - {2^{ - \frac{{B_{k}^{}}}{{{M-1}}}}}\ln 2\frac{{{a_{k}}}}{{{M-1}}} + \nu  = 0,
\end{equation}
where $\nu$ is the Lagrange multiplier and  $\rm U$ is the set of feedback links ${\rm{U}} = \left\{ {1,...,K} \right\}$. From \eqref{Eq:IA:L_0}, we obtain $B_{k}$ as
\begin{equation}\label{Eq:IA:B_solution}
B_{k}^{} = {(M-1)} \cdot {\log _2}\left( {\frac{{\mu {a_{k}}}}{{{M-1}}}} \right)
\end{equation}
under the following constraint
\begin{equation}\label{Eq:IA:M_ST}
\sum\limits_{k \in {\mathop{\rm U}\nolimits} }^{} {{(M-1)} \cdot{{\log }_2}\left( {\frac{{\mu {a_{k}}}}{{{M-1}}}} \right)}  = {B_{T}}
\end{equation}
where $\mu  = \frac{{\ln 2}}{\upsilon }$. Combining \eqref{Eq:IA:B_solution} and \eqref{Eq:IA:M_ST} with  $B_{k}^{} \ge 0$, the number of optimal feedback-bit $B_{k}^*$ that minimizes the sum residual interference is obtained as\footnote{ ${\left( a \right)^ + }$ implies that
 ${\left( a \right)^ + } = \left\{ \begin{array}{l}
 a\,\,\,\,if\,\,a \ge 0 \\ 
 0\,\,\,\,if\,\,a < 0 \\ 
 \end{array} \right.$.}
\begin{equation}\label{Eq:IA:B_opt}
B_k^* = \frac{1}{{\left| {\mathop{\rm U}\nolimits}  \right|}}{\left( {\gamma  - (M-1){\cdot}{\left| {\mathop{\rm U}\nolimits}  \right|}{\cdot}{{\log }_2}\left( {\frac{{M - 1}}{{{a_k}}}} \right)} \right)^ + }
\end{equation}
where $|$U$|$ denotes the cardinality of U and $\gamma  = {B_{T}} + \,\,\sum\limits_{k \in {\mathop{\rm U}\nolimits} } {{(M-1)} \cdot {{\log }_2}\left( {\frac{{{M-1}}}{{{a_{k}}}}} \right)}$.  The solution of \eqref{Eq:IA:B_opt} is  found  through the waterfilling algorithm, described in Algorithm \ref{waterfilling}. 

\begin{algorithm}[t]\label{waterfilling}
{\bf 1.  Waterfilling Solution }

$i$=$0$;

${\rm{U}} = \left\{ {1,...,K} \right\}$;

\While{$i$=$0$}{
Determine the water-level $\gamma  = {B_{T}} + \,\,\sum\limits_{k \in {\mathop{\rm U}\nolimits} } {{(M-1)} \cdot {{\log }_2}\left( {\frac{{{M-1}}}{{{a_{k}}}}} \right)}$.

Choose the user set $ k^{*} = \arg \max \left\{ {\frac{{{M-1}}}{{{a_{k}}}}:k \in {\mathop{\rm U}\nolimits} } \right\}$.

\lIf{$\gamma  - (M-1){\cdot}{\left| {\mathop{\rm U}\nolimits}  \right|}{\cdot}{\log _2}\left( {\frac{{{M-1}}}{{{a_{ k^{*}}}}}} \right) \ge 0$}  {$\left\{ {B_{k}^*: k \in {\rm{U}}} \right\}$ is determined by \eqref{Eq:IA:B_opt}.

$i$=$i$+$1$;

}
\lElse{ 
Let define
${{\mathop{\rm U}\nolimits} _{}} = \left\{ {{{\mathop{\rm U}\nolimits} _{}}\,\,{\rm{except \,\, for}}\,\, k^{*}} \right\}$ and $B_{ k^{*}}^* = 0$.}
}
{
{\bf 2. Decision of Feedback Bits}
Under the integer constant, the optimal feedback-bits $\left\{ {B_{k}^*: k \in {\rm{U}}} \right\}$ are rounded as
\begin{equation}\label{Eq:IA:B_int}
 B_k^* = \left\lfloor {B_k^*} \right\rfloor
\end{equation}
that satisfies $\sum\nolimits_{k=1}^{K}  B_k^*=B_{T}$,   where $\left\lfloor {x} \right\rfloor$ is the largest integer not greater than $x$.}

\caption{Waterfilling Algorithm for Dynamic Feedback-Bit Allocation}
\end{algorithm} 
{As can be seen in Fig. \ref{Fig:Controller:Bit}, the optimal number of feedback-bits $B_k^{*}$ is allocated over the inverse of the interference channel gains due to the quantization error of ${\bf v}^{[k]}$, until it does not overflow the water-level chosen to satisfy the constraint of $B_T$.}

% \subsection{Path-loss based feedback-bits allocation strategy}
{
\begin{itemize}
\item \emph {Overhead of the dynamic feedback-bits allocation} : The computation of $\left\{ {B_{k}^*: k \in {\rm{U}}} \right\}$ in \eref{Eq:IA:B_opt} requires the set of interfering channel gains $\left\{ {a_k} \right\}_{k = 1}^K$ at the centralized bit controller, which consists of the variability of path-loss between cross-links. Since  the dynamic feedback-bits allocation is performed by  gathering a long-term CSI  which represents a slow variability compared with small-scale fading channel, it does not require the frequent CS-exchange between receivers and the controller. This provides the benefit of lower cost for implementation.
\end{itemize}}

\subsection{Scaling Law of Total Feedback Bits}
 Each pair of transmitter-receiver obtains the interference-free link for its desired data stream under the IA strategies. However, misaligned beamformers in limited feedback channel destroy the linear scaling gain of sum capacity, especially at high SNR regime. In this subsection, we analyze the total number of feedback-bits that achieves the same DoF as the case of perfect CSI in the proposed feedback topologies. To achieve linear scaling DoF in $K$-user interference channel, the sum of residual interference maintains the constant value over whole SNR regimes according to
\begin{equation}\label{Eq:IA:DoF}
\begin{array}{l}
 DoF = \mathop {\lim }\limits_{P \to \infty } \frac{E_{\bf H, \mathcal{W}}[{R_{\textsf{limited}}}]}{{{{\log }_2}P}} \\ 
 \,\,\,\,\,\,\,\,\,\,\,\,\,\,\, = \mathop {\lim }\limits_{P \to \infty } \frac{{\sum\limits_{k = 1}^K {E_{\bf H, \mathcal{W}}\left [{\log _2}\left(P{d_{kk}^{-\alpha}{{\left| {{{{\bf{\hat r}}}^{[k]}}{{\bf{H}}^{[kk]}}{{{\bf{\hat v}}}^{[k]}}} \right|}^2}} \right) \right]}  - \sum\limits_{i = 1}^K {E_{\bf H, \mathcal{W}} \left [ {{\log }_2}\left( {{{\hat I}^{[k]}}} \right) \right ]} }}{{{{\log }_2}P}} \\ 
\,\,\,\,\,\,\,\,\,\,\,\,\,\,\, = \mathop {\lim }\limits_{P \to \infty } \frac{{\sum\limits_{k = 1}^K E_{\bf H, \mathcal{W}} \left[ {{\log _2}\left(P{d_{kk}^{-\alpha} {{\left| {{{{\bf{\hat r}}}^{[k]}}{{\bf{H}}^{[kk]}}{{{\bf{\hat v}}}^{[k]}}} \right|}^2}} \right)}\right ] }}{{{{\log }_2}P}} - \mathop {\lim }\limits_{P \to \infty } \frac{{\sum\limits_{k = 1}^K E_{\bf H, \mathcal{W}}\left[{{{\log }_2}\left( {{{\hat I}^{[k]}}} \right) } \right] }}{{{{\log }_2}P}} \\ 
\,\,\,\,\,\,\,\,\,\,\,\,\,\,\, \ge K- \mathop {\lim }\limits_{P \to \infty } \frac{{\sum\limits_{k = 1}^K {{{\log }_2}\left( {{E_{\bf H, \mathcal{W}} \left[ {\hat I}^{[k]}\right]}} \right)} }}{{{{\log }_2}P}} \\ 
 \,\,\,\,\,\,\,\,\,\,\,\,\,\,\, \mathop  = \limits^{(a)} K \\ 
 \end{array}
 \end{equation}
where (a) follows from the constant value of $\sum\nolimits_{k = 1}^K {{{\log }_2}\left( E_{\bf H, \mathcal{W}}{{\left[{\hat I}^{[k]}\right]}} \right)}$.

Let denote $B_T^*$ as the total feedback bits that achieve linearly scaling DoF with $K$. Then we formulate the sum residual interference as the function of $a_k$ and $B_k$
\begin{equation}\label{Eq:IA:B_sum_const}
\begin{array}{l}
 \sum\limits_{k = 1}^K {{{\log }_2}\left( E_{\bf H, \mathcal{W}}\left[ {{{\hat I}^{[k]}}}\right] \right)} \le \sum\limits_{k = 1}^K {\log _2}\left( {{{{ a}_k}{2^{ - \frac{{{B_k}}}{{M - 1}}}}} } \right) \\ 
 \,\,\,\,\,\,\,\,\,\,\,\,\,\,\,\,\,\,\,\,\,\,\,\,\,\,\, \,\,\,\,\,\,\,\,\,\,\,\,\,\,\,\,\,\,\,\,\,\,\,\,\,\,\,\, = \sum\limits_{k = 1}^K {\log _2}\left( {{ a}_k}\right) + \sum\limits_{k = 1}^K {\log _2}\left( {{{}{2^{ - \frac{{{B_k}}}{{M - 1}}}}} } \right) \\  
\,\,\,\,\,\,\,\,\,\,\,\,\,\,\,\,\,\,\,\,\,\,\,\,\,\,\, \,\,\,\,\,\,\,\,\,\,\,\,\,\,\,\,\,\,\,\,\,\,\,\,\,\,\,\, =C\\ 
 \end{array}
 \end{equation}
 where $C>0$ is constant. 
From \eqref{Eq:IA:B_sum_const}, we obtain $B^*_T$ as
\begin{equation}\label{Eq:IA:B_sum_opt}
\begin{array}{l}
 B_{T}^* = \sum\limits_{k = 1}^K {{B_k}}  \\ 
 \,\,\,\,\,\,\,\,\,\, = \left( {M - 1} \right) \cdot \left( {\sum\limits_{k=1}^K {{{\log }_2}} \,{{ a}_k} - C} \right) \\ 
 \,\,\,\,\,\,\,\,\,\, = K \cdot \left( {M - 1} \right) \cdot {\log _2}P + \left( {M - 1} \right)\cdot\left( {\sum\limits_{k = 1}^K {{{\log }_2}} \,{{\hat a}_k} - C} \right) \\ 
 \end{array}
\end{equation}
where ${\hat a}_k=\frac{{ a}_k}{P}$. Since the total feedback-bits is the non-negative integer, we determine ${ B}_{T}^*$ as 
\begin{equation}\label{Eq:IA:B_sum_bar}
 B_{T}^* = {\rm{nint}}(B_{T}^*)
\end{equation}
where ${\rm{nint}}(\emph {x})$ is the nearest integer function of \emph{x}.

 {
 \begin{remark} Compared with the total feedback bits in the full-feedback topology  in [15]  that scale with $\mathcal O \left(K^2\cdot (M^2-1)\cdot \log_2P \right)$, the proposed feedback topologies requires $B_T^{*}=\mathcal O \left(K\cdot(M^2-1)\cdot \log_2P\right)$ bits to achieve the full DoF. The smaller number of feedback bits results from the proposed feedback structure that sequentially computes the IA precoder based on the pre-determined ones and exchanges precoders between transmitters and receivers.
 \end{remark}
 }

 The required total feedback bits \eqref{Eq:IA:B_sum_bar} that achieve $K$ DoF can not be computed without the centralized controller. To implement the feedback-bits allocation in the distributed $K$-user networks, we modify the constraint in \eref{Eq:IA:B_sum_const} into the individual constraints of receiver $k$ as follows.
\begin{equation}\label{Eq:IA:E_b_k_cond}
\begin{array}{l}
 {\log _2}\left( \mathbb {E}_{ {\bf H}, \mathcal{W}}\left[{{{\hat { I}}^{[k]}}}\right] \right) \le {\log _2}\left( {{a_k}{2^{ - \frac{{{B_k}}}{{M - 1}}}}} \right)  = \frac{C}{K}, \,\,\, \forall k. \\ 
 \end{array}
\end{equation}
Then, the required feedback bits ${ B}_k^{d^*}$ for the $k$-th beamformer that achieve $K$ DoF are derived as
\begin{equation}\label{Eq:IA:E_b_k_distributed}
\begin{array}{l}
 B _k^{d^*}={\rm{nint}}(B_{k}^d)
\\
\end{array}
\end{equation}
where ${B_k^d} = \left( {M - 1} \right) \cdot \left ( {\log _2}{ a}_k -{\frac{C}{K}} \right)$.  From the components that consist of $a_k$ in \eref{Eq:IA:a_k_m},  ${B}_{K-1}^{d^*}$ and ${B}_{K}^{d^*}$ are computed under the assumption of long-term CSI exchange between receiver $(K-1)$ and $K$ and  other ${B}_k^{d^*}$ is computed with local channel knowledge at receiver $(k-2)$. Based on \eqref{Eq:IA:E_b_k_distributed}, we design the CSI-Exchange topology with the distributive feedback-bits allocation, which achieves linearly scaling DoF with $K$ in limited feedback channel. The procedure is represented in Algorithm \ref{Exchange_2}.

\begin{algorithm} [t]\label{Exchange_2}
\,\,
\textbf{1. Feedback of  ${\bf{\hat v}}^{[1]}$ and ${\bf{\hat v}}^{[2]}$}: Receiver $K$ computes ${\bf{ v}}^{[k]}$ and the corresponding feedback bits ${ B}_k^{d^*}$, where $k=1, 2$. Then, ${\bf{ v}}^{[k]}$ is quantized to ${\bf{\hat v}}^{[k]}$ using ${ B}_k^{d^*}$ bits random codebook and fed back to the corresponding transmitter $k$. 

\textbf{2. Feedback of ${\bf{\hat v}}^{[3]} ,...,{\bf{\hat v}}^{[K]}$} \\
\, \For{$k$=$3$:$K$}{
 $(k-2)$-th receiver calculates ${B}_k^{d^*}$ feedback bits and ${ {\bf {v}}}^{[k]}$. Then,  ${ {\bf {v}}}^{[k]}$ is quantized to ${\hat {\bf {v}}}^{[k]}$ and  fed back to the transmitter $k$.\\
}
\,\,\,
\caption{CSI-Exchange Topology with Distributive Feedback-bits Design}
\end{algorithm}
\,\,

\section{Numerical Results}
In this section, we represent the performance of proposed IA feedback topologies in $K$-user MIMO interference channel with limited feedback. The throughput improvement of dynamic feedback-bits allocation in the different feedback topologies is verified comparing with the conventional case of \emph {equal feedback-bit allocation} in which the number of feedback-bits sent to each transmitter is equal and fixed. Also, the scaling law of total feedback-bits that achieves same DoF as the case of perfect CSI is shown in both centralized and distributed $K$-user network models. 
\begin{figure}[t]
\centering\includegraphics[width=13cm]{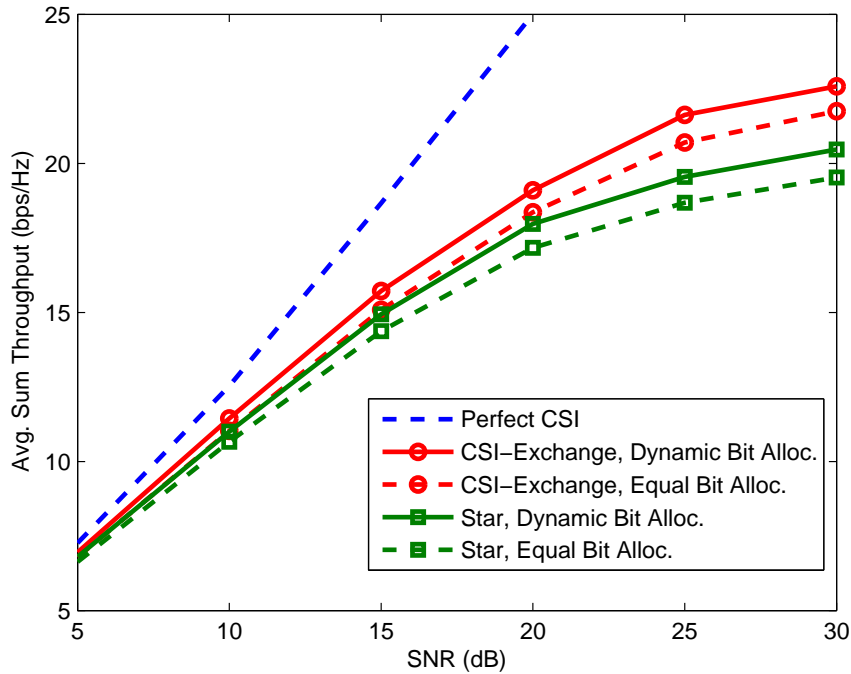}
\caption{Comparison of average sum throughput between the star and CSI-exchange topology in $B_{T}$=$16$.}
\label{Fig:Limited_Feedback:Performance_16}
\end{figure} 

\begin{figure}[t]
\centering\includegraphics[width=13cm]{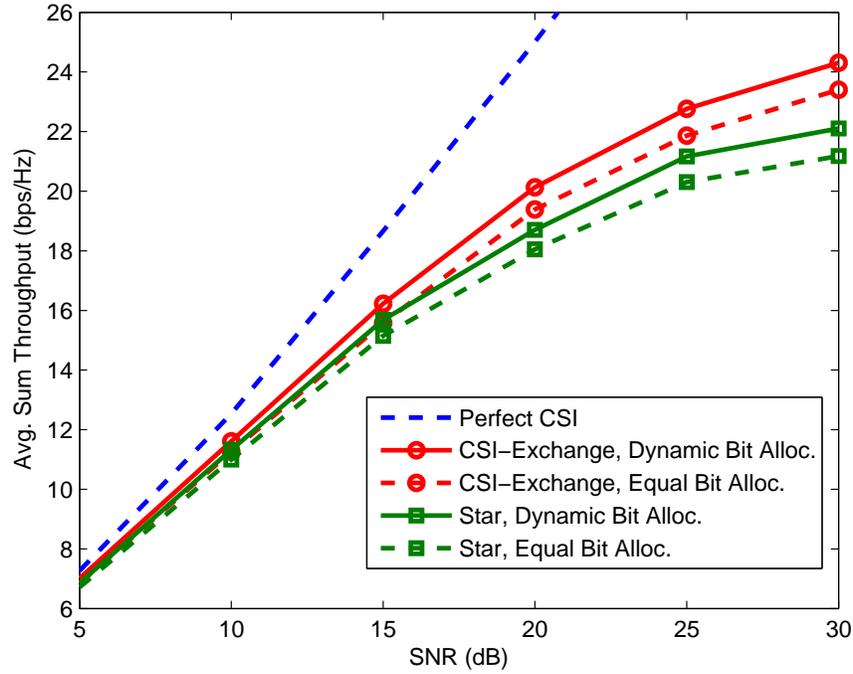}
\caption{Comparison of average sum throughput between the star and CSI-exchange topology in $B_{T}$=$20$.}
\label{Fig:Limited_Feedback:Performance_20}
\end{figure} 
Fig. \ref{Fig:Limited_Feedback:Performance_16} and Fig. \ref{Fig:Limited_Feedback:Performance_20} show the average sum throughput of $4$-user $3$$\times$$3$ MIMO IA channel in limited feedback channel, which are constructed by the star and CSI-exchange feedback topology with the dynamic feedback-bits algorithm proposed in section V.A. For comparison, the curve for equal feedback-bit allocation is also plotted. { We set the parameters $B_{T}=16$ and $20$,  $\alpha=3.5$ and assume that each pair of transmitter-receiver is uniformly distributed within the range of ${\left( {{d_{kj}}/{d_{kk}}} \right)_{k \ne j}}\in[$$1$  $3$$]$.}
Moreover, the small-scale fading is assumed to be static during the procedure of proposed feedback algorithms.
While the dynamic feedback-bits allocation scheme requires the centralized feedback-bits controller, it provides the performance enhancement than the equal feedback-bits allocation scheme. The performance gap between dynamic and equal bit allocation becomes larger in high SNR since the proposed feedback-bits allocation effectively regulates the strong residual interference.

{
Compared the performance of CSI-exchange with that of star feedback topology in the same feedback bits, the CSI-exchange topology always provides a higher sum throughput than the star feedback topology. This is due to the different procedure of precoder design in limited feedback environment, represented in section IV.B. In the star feedback topology, all precoders are simultaneously computed at CSI-BS and fed back to the corresponding transmitters through the limited feedback channel, which causes two misaligned interferers at every receiver. However, all receivers except for receiver $(K-1)$ and $K$ experience a single misaligned interferers in the CSI-exchange topology, since IA precoders  are sequentially designed on the subspace of pre-determined quantized precoder.}

In Fig. \ref{Fig:Limited_Feedback:Performance_16} and \ref{Fig:Limited_Feedback:Performance_20}, we find that a linearly scaling DoF can not be achieved in high SNR regimes since the given $B_T$ feedback-bits are not enough to manage the strong residual interference. In \eref{Eq:IA:B_sum_bar} and \eref{Eq:IA:E_b_k_distributed}, we suggest that the required number of feedback bits that achieves linear scaling law of sum throughput in both centralized and distributed networks. Fig.\ref{Fig:Limited_Feedback:M_CSI_IA} shows the performance of sum throughput under the following two feedback-bits allocation schemes based on \eqref{Eq:IA:B_sum_bar} and \eref{Eq:IA:E_b_k_distributed}:
\begin{itemize}
 \item[1)] Centralized feedback-bits allocation strategy : It assumes that the centralized controller collects all cross-link gains $\left\{ {a_k} \right\}_{k = 1}^K$ and computes $B^{*}_T$ using \eqref{Eq:IA:B_sum_opt} and \eqref{Eq:IA:B_sum_bar}. Each of ${ B}_k^*$ is computed by the dynamic feedback-bits allocation scheme in Algorithm \ref{waterfilling}.
  \item[2)] Distributed bit allocation strategy: Each ${ B}_k^{d^*}$ is sequentially computed based on the local CSI at receiver $k$. Details are described in Algorithm \ref{Exchange_2}.
 \end{itemize}
Setting the parameters $C=2$ and  ${\left( {{d_{kj}}/{d_{kk}}} \right)_{k \ne j}} = 2$, the sum throughput in both two schemes shows linear increase over the whole SNR regions and the total feedback-bits in both 1) and 2) are logarithmically increased with $P$.

\begin{figure}[t]
\centering\includegraphics[width=12cm]{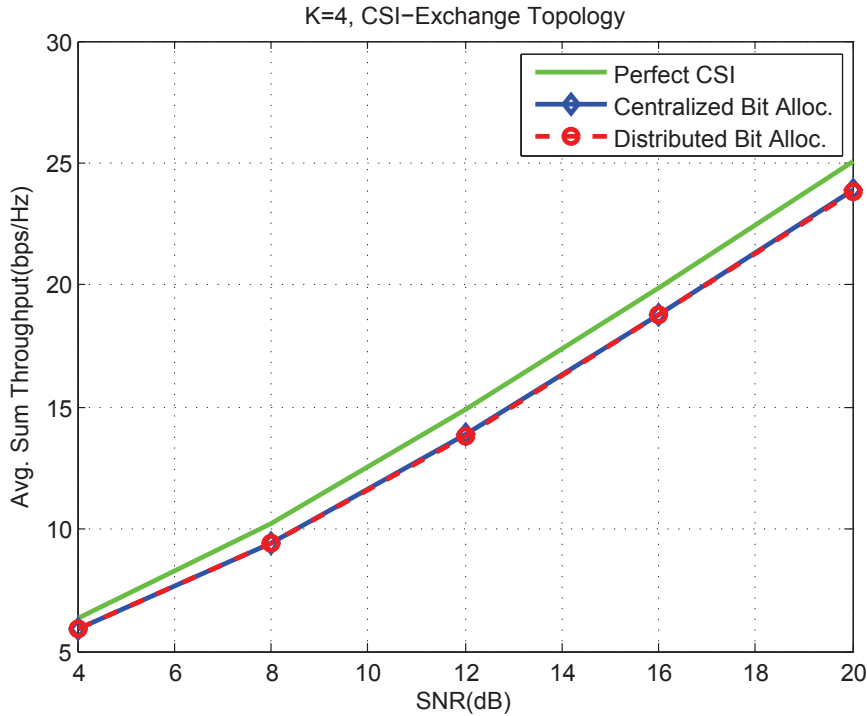}
\caption{Scaling law of sum throughput in limited feedback channel.}
\label{Fig:Limited_Feedback:M_CSI_IA}
\end{figure}

\section{Conclusion}
In this paper, the efficient feedback topologies for IA have been proposed in $K$-user MIMO interference channels. We showed that the proposed feedback topologies provide the dramatic reduction of network overhead compared with a conventional feedback framework for IA. In the context of limited feedback channel, we analyzed the upper bounds of sum residual interference in given feedback topologies which are equivalent to the sum throughput loss due to the quantization error. Using these bounds, we suggested the dynamic feedback-bits allocation scheme that minimizes sum residual interference with the water filling solution. The performance gain of dynamic feedback-bits allocation was dominant in high SNR regions while each of receiver is affected by a strong residual interference. Furthermore, the scaling law of feedback bits achieving IA is derived, which is linearly increased with $K(M-1)$ and ${\log}_{2}P$. For the practical implementation, we developed the feedback-bits allocation scheme that only requires the long-term channel gains and distributive feedback-bits controlling system without any centralized controller.   

\appendices

\section{Proof of Proposition~\ref{Prop:STAR:E_SUM_I}} \label{1}

In the centralized-feedback topology, the design of IA precoders follows \eqref{Eq:IA:GS}. From the IA condition~\eref{Eq:IA:ex_cond}, the interference from transmitter $\hat k$ and $\bar k$ is aligned on the reference vector {${\bf v}_r^{[k]}={{{\bf{H}}^{[k \hat k]}}} {\bf{ v}}_{}^{[\hat k]}={{{\bf{H}}^{[k \bar k]}}} {\bf{ v}}_{}^{[\bar k]}$, $\forall k$.} However, the finite-rate feedback channel from receiver 1 (CSI-BS) to the corresponding transmitters causes the quantization error of beamformers so that the residual interference is generated by both quantized ${\bf{\hat v}}^{[\hat k]}$ and  ${\bf{\hat v}}^{[\bar k]}$ at receiver $k$. Then, we derive ${\hat I}^{[k]}$ as \footnote{To derive the residual interference, the set of $\left\{ {{{{\bf{\hat v}}}^{[k]}}, \forall k} \right\}$ and ${\bf v}_r^{[k]}$ are assumed to be known at receiver $k$.}
\begin{align}
{{{{\hat I}}}^{[k]}} &= \sum\limits_{m = 1, m \ne k}^KPd_{km}^{ - \alpha}{\left| {{{{{\bf{\hat r}}}^{[k]\dag }}{{\bf{H}}^{[km]}}{{{\bf{\hat v}}}^{[m]}}} } \right|^2} \nonumber\\
&= Pd_{k\hat k}^{ - \alpha }{\left| {{{{\bf{\hat r}}}^{[k]\dag }}{{\bf{H}}^{[k\hat k]}}{{{\bf{\hat v}}}^{[\hat k]}}} \right|^2} +Pd_{k \bar k}^{ - \alpha }{\left| {{{{\bf{\hat r}}}^{[k]\dag }}{{\bf{H}}^{[k \bar k]}}{{{\bf{\hat v}}}^{[\bar k]}}} \right|^2} \label{Eq:STAR:I_k}\\
&= Pd_{k\hat k}^{ - \alpha }\sigma_{\hat k}{\left| {{{{\bf{\hat r}}}^{[k]\dag }}{{\bf{H}}^{[k\hat k]}}{{{\Delta \bf{ v}}}^{[\hat k]}}} \right|^2} +Pd_{k \bar k}^{ - \alpha }\sigma_{\bar k}{\left| {{{{\bf{\hat r}}}^{[k]\dag }}{{\bf{H}}^{[k \bar k]}}{{{\Delta \bf{ v}}}^{[\bar k]}}} \right|^2} \nonumber\\
&\mathop  \le \limits^{(a)} Pd_{k\hat k}^{ - \alpha }\sigma _{\hat k}{\left\| {{{{\bf{\hat r}}}^{[k]\dag }}{{\bf{H}}^{[k \hat k]}}} \right\|^2}{\left\| {\Delta {{\bf{v}}^{[\hat k]}}} \right\|^2} + Pd_{k \bar k}^{ - \alpha }\sigma _{\bar k}{\left\| {{{{\bf{\hat r}}}^{[k]\dag }}{{\bf{H}}^{[k \bar k]}}} \right\|^2}{\left\| {\Delta {{\bf{v}}^{[\bar k]}}} \right\|^2}\nonumber
\end{align}

where { ${{\bf{\hat r}}^{[k]}}$ is designed to lie in the nullspace of $\{{\bf v}_{r}^{[k]}, {\bf \hat I}^{[k]}\} $, ${\bf \hat I}^{[k]}=\{   {\bf{H}}^{[km]}{\hat {\bf v}}^{[m]} | \,  \forall \,m, m \ne \hat k, \bar k, k \} $} and 
(a) follows from Cauchy-Schwarz inequality.\footnote{Note that ${\left| {{{\bf{a}}^\dag }{\bf{b}}} \right|^2} \le {\left\| {\bf{a}} \right\|^2}{\left\| {\bf{b}} \right\|^2}$, where ${\bf{a}},{\bf{b}} \in {\mathbb{C}^{M \times 1}}$}

% ${{\bf{\hat r}}^{[k]}}$ is designed on the nullspace of $\{{\bf{H}}^{[k1]}{\hat {\bf v}}^{[1]}, \cdots, {\bf{H}}^{[kk-1]}{\hat {\bf v}}^{[k-1]}, {\bf v}_{r}^{[k]},{\bf{H}}^{[k \bar k+1]}{\hat {\bf v}}^{[\bar k+1]}, \cdots {\bf{H}}^{[kK]}{\hat {\bf v}}^{[K]}\} $
 From this result, the expectation of ${{{\hat I}^{[k]}}}$ is upper-bounded by
\begin{align}
{\mathbb E_{{\bf H}, \mathcal{W}}}\left[ {{{\hat I}^{[k]}}} \right]  &\le  Pd_{k\hat k}^{ - \alpha } {\mathbb E} \left[\sigma _{\hat k}\right]{\mathbb E_{\bf H}}\left[{\left\| {{{\bf{H}}^{[k \hat k]}}} \right\|^2}\right]+ Pd_{k \bar k}^{ - \alpha }\mathbb E \left[\sigma _{\bar k}\right] {\mathbb E_{\bf H}}\left[{\left\| {{{\bf{H}}^{[k \bar k]}}} \right\|^2}\right ] \nonumber \\
&\le {\bar \Gamma(M)}\cdot  \left(Pd_{k \hat k}^{ - \alpha } \cdot M^2 \cdot 2^{- {B_{\hat k}\over {M-1}}} + Pd_{k \bar k}^{ - \alpha }\cdot M^2 \cdot 2^{- {B_{\bar k}\over {M-1}}}\right)    , \,\,\, \forall k \label{Eq:STAR:E_I_k}.
\end{align}

\section{Proof of Proposition~\ref{Prop:H:E_SUM_I}} \label{2}

Consider the residual interference at receiver $(K-1)$ and $K$ affected by a misalignment between the interference from transmitter $1$ and $2$.  The ${{{\hat I}}^{[K-1]}}$ is derived as
\begin{align}
{{{{\hat I}}}^{[K-1]}} &= \sum\limits_{k = 1, k \ne (K-1)}^KPd_{(K-1)k}^{ - \alpha}{\left| {{{{{\bf{\hat r}}}^{[K-1]\dag }}{{\bf{H}}^{[(K-1)\,k]}}{{{\bf{\hat v}}}^{[k]}}} } \right|^2}  \nonumber\\
&= Pd_{(K-1)1}^{ - \alpha }{\left| {{{{\bf{\hat r}}}^{[K-1]\dag }}{{\bf{H}}^{[(K-1)\,1]}}{{{\bf{\hat v}}}^{[1]}}} \right|^2} +Pd_{(K-1)\, 2}^{ - \alpha }{\left| {{{{\bf{\hat r}}}^{[K-1]\dag }}{{\bf{H}}^{[(K-1)\, 2]}}{{{\bf{\hat v}}}^{[2]}}} \right|^2} \label{Eq:ex:I_K_1} \\
&= Pd_{(K-1)1}^{ - \alpha }\sigma_{1}{\left| {{{{\bf{\hat r}}}^{[K-1]\dag }}{{\bf{H}}^{[(K-1)\,1]}}{{{\Delta \bf{ v}}}^{[1]}}} \right|^2} +Pd_{(K-1) 2}^{ - \alpha }\sigma_{2}{\left| {{{{\bf{\hat r}}}^{[K-1]\dag }}{{\bf{H}}^{[(K-1)\, 2]}}{{{\Delta \bf{ v}}}^{[2]}}} \right|^2}\nonumber \\
& \mathop \le  Pd_{(K-1)1}^{ - \alpha }\sigma _{1}{\left\| {{{{\bf{\hat r}}}^{[K-1]\dag }}{{\bf{H}}^{[(K-1)1]}}} \right\|^2}{\left\| {\Delta {{\bf{v}}^{[1]}}} \right\|^2} + Pd_{(K-1) 2}^{ - \alpha }\sigma _{2}{\left\| {{{{\bf{\hat r}}}^{[K-1]\dag }}{{\bf{H}}^{[(K-1)2]}}} \right\|^2}{\left\| {\Delta {{\bf{v}}^{[2]}}} \right\|^2} \nonumber
\end{align}
where {${{\bf{\hat r}}^{[K-1]}}$ is on the nullspace of $\{  {\bf v}_r^{[K-1]}, {\bf{H}}^{[(K-1)\,3]}{\hat {\bf v}}^{[3]}, \cdots, {\bf{H}}^{[(K-1)\,(K-2)]}{\hat {\bf v}}^{[K-2]}, {\bf{H}}^{[(K-1)\, K]}{\hat {\bf v}}^{[K]}\} $.} Then, the expectation of ${{{\hat I}}^{[K-1]}}$ is upper bounded by
\begin{align}\label{Eq:Eex:I_K_1}
{\mathbb E_{{\bf H}, \mathcal{W}}}\left[ {{{\hat I}^{[K-1]}}} \right]&\le Pd_{(K-1)1}^{ - \alpha } {\mathbb E} \left[\sigma _{1}\right]{\mathbb E_{\bf H}}\left[{\left\| {{{\bf{H}}^{[(K-1)1]}}} \right\|^2}\right]+ Pd_{(K-1)  2}^{ - \alpha }\mathbb E \left[\sigma _{2}\right] {\mathbb E_{\bf H}}\left[{\left\| {{{\bf{H}}^{[(K-1) 2]}}} \right\|^2}\right ]\nonumber \\
&\le {\bar \Gamma(M)}\cdot  \left(Pd_{(K-1) 1}^{ - \alpha } \cdot M^2 \cdot 2^{- {B_{1}\over {M-1}}} + Pd_{(K-1) 2}^{ - \alpha }\cdot M^2 \cdot 2^{- {B_{2}\over {M-1}}} \right).
\end{align}
Similarly, { ${{\bf{\hat r}}^{[K]}}$ is designed to lie in the nullspace of $\{  {\bf v}_r^{[K]}, {\bf{H}}^{[K3]}{\hat {\bf v}}^{[3]}, \cdots, {\bf{H}}^{[K (K-1)]}{\hat {\bf v}}^{[K-1]}\}$} and then ${{{\hat I}}^{[K]}}$ is derived as
\begin{align}
{{{{\hat I}}}^{[K]}} &= \sum\limits_{k = 1, k \ne K}^KPd_{Kk}^{ - \alpha}{\left| {{{{{\bf{\hat r}}}^{[K]\dag }}{{\bf{H}}^{[Kk]}}{{{\bf{\hat v}}}^{[k]}}} } \right|^2} \nonumber\\
&= Pd_{K1}^{ - \alpha }{\left| {{{{\bf{\hat r}}}^{[K]\dag }}{{\bf{H}}^{[K1]}}{{{\bf{\hat v}}}^{[1]}}} \right|^2} +Pd_{K2}^{ - \alpha }{\left| {{{{\bf{\hat r}}}^{[K]\dag }}{{\bf{H}}^{[K 2]}}{{{\bf{\hat v}}}^{[2]}}} \right|^2}\label{Eq:ex:I_K}\\
&= Pd_{K1}^{ - \alpha }\sigma_{1}{\left| {{{{\bf{\hat r}}}^{[K]\dag }}{{\bf{H}}^{[K1]}}{{{\Delta \bf{ v}}}^{[1]}}} \right|^2} +Pd_{K 2}^{ - \alpha }\sigma_{2}{\left| {{{{\bf{\hat r}}}^{[K]\dag }}{{\bf{H}}^{[K 2]}}{{{\Delta \bf{ v}}}^{[2]}}} \right|^2}\nonumber\\
&\mathop  \le  Pd_{K1}^{ - \alpha }\sigma _{1}{\left\| {{{{\bf{\hat r}}}^{[K]\dag }}{{\bf{H}}^{[K 1]}}} \right\|^2}{\left\| {\Delta {{\bf{v}}^{[1]}}} \right\|^2} + Pd_{K2}^{ - \alpha }\sigma _{2}{\left\| {{{{\bf{\hat r}}}^{[K]\dag }}{{\bf{H}}^{[K 2]}}} \right\|^2}{\left\| {\Delta {{\bf{v}}^{[2]}}} \right\|^2}. \nonumber
\end{align}

Then, the upper-bound of expectation of ${\mathbb E_{{\bf H}, \mathcal{W}}}\left[ {{{\hat I}^{[K]}}} \right]$ is represented as 
\begin{align}\label{Eq:ex:I_K}
{\mathbb E_{{\bf H}, \mathcal{W}}}\left[ {{{\hat I}^{[K]}}} \right]&\le \mathbb E_{\bf H, \mathcal W} \left[ Pd_{K1}^{ - \alpha }{\left| {{{{\bf{\hat r}}}^{[K]\dag }}{{\bf{H}}^{[K1]}}{{{\bf{\hat v}}}^{[1]}}} \right|^2} +Pd_{K 2}^{ - \alpha }{\left| {{{{\bf{\hat r}}}^{[K]\dag }}{{\bf{H}}^{[K 2]}}{{{\bf{\hat v}}}^{[2]}}} \right|^2}\right ]\nonumber \\
&\le{\bar \Gamma(M)}\cdot  \left( Pd_{K 1}^{ - \alpha } \cdot M^2 \cdot 2^{- {B_{1}\over {M-1}}} + Pd_{K 2}^{ - \alpha }\cdot M^2 \cdot 2^{- {B_{2}\over {M-1}}}\right). 
\end{align}

From Algorithm \ref{Exchange_1}, the ${{\bf{ v}}^{[k+2]}}$ is sequentially designed on the subspace of ${\bf v}_{r}^{[k]}={\bf H}^{[k\, (k+1)]}{\bf \hat v}^{[k+1]}$ and fed back to transmitter $(k+2)$ through $B_{k+2}$ bits feedback channel at receiver $k$, $k=1,..., K-2$. Therefore, the quantized ${{\bf{\hat v}}^{[k+2]}}$ causes the misalignment with ${\bf v}_{r}^{[k]}$ which generates the residual interference at receiver $k$. Then, we derive ${{{\hat I}}^{[k]}}$ as
\begin{align}
 {{{{\hat I}}}^{[k]}} &= \sum\limits_{m = 1, m \ne k}^KPd_{km}^{ - \alpha}{\left| {{{{{\bf{\hat r}}}^{[k]\dag }}{{\bf{H}}^{[km]}}{{{\bf{\hat v}}}^{[m]}}} } \right|^2} \nonumber\\ 
&= Pd_{k \,(k+2)}^{ - \alpha }{\left| {{{{\bf{\hat r}}}^{[k]\dag }}{{\bf{H}}^{[k\, (k+2)]}}{{{\bf{\hat v}}}^{[k+2]}}} \right|^2} \label{Eq:ex:I_k}\\
&= Pd_{k\, (k+2)}^{ - \alpha }\sigma_{k+2}{\left| {{{{\bf{\hat r}}}^{[k]\dag }}{{\bf{H}}^{[k \,(k+2)]}}{{{\Delta \bf{ v}}}^{[k+2]}}} \right|^2} \label{Eq:ex:I_k} \nonumber\\
&\mathop  \le  Pd_{k\, (k+2)}^{ - \alpha }\sigma _{k+2}{\left\| {{{{\bf{\hat r}}}^{[k]\dag }}{{\bf{H}}^{[k \,(k+2)]}}} \right\|^2}{\left\| {\Delta {{\bf{v}}^{[k+2]}}} \right\|^2}  \nonumber
\end{align}
where { ${{\bf{\hat r}}^{[k]}}$ is on the nullspace of $\{{\bf v}_{r}^{[k]}, {\bf \hat I}^{[k]}\} $ and ${\bf \hat I}^{[k]}=\{   {\bf{H}}^{[km]}{\hat {\bf v}}^{[m]} | \,  \forall\, m, m \ne k,k+1,k+2 \}$}. From this result, the expectation of ${\hat I}^{[k]}$  is upper bounded by
\begin{equation}\label{Eq:ex:I_k_bound}
{E_{\bf H, \mathcal{W}}}\left[ {{{\hat I}^{[k]}}} \right] \le {\bar \Gamma(M)}\cdot  Pd_{k\,(k+2)}^{ - \alpha }\cdot M^2 \cdot 2^{- {B_{k+2} \over {M-1}}}.
\end{equation}

\section{Proof of Lemma~\ref{Lemma:H:Q_V1}} \label{3} {
Note that the inaccurate ${\bf \bar v}^{[1]}$ due to the quantization error of  ${\bf{ H}}_{e}^{[K-1]}$ is computed as
\begin{equation}\label{Eq:IA:M_V_1_K_Quantize}
\begin{array}{l}
 {\bf{\bar v}}_{}^{[1]}\,\,{\rm{ = }}\,{\textsf{eigenvector \,of}}\,\left( {{\bf{\hat H}}_{e}^{[K-1]}}{{\bf{ H}}_{e}^{[K]}}    \right) \\ 
 \,\,\,\,\,\,\,\,\,\,={\textsf{eigenvector \,of}} \left(\sqrt {1-\sigma_{M}} \, {{\bf{H}}_{e}^{[K-1]}}{\bf{H}}_{e}^{[K]}\, + \sqrt {\sigma_{M}} {\Delta {\bf{H}}_{e}^{[K-1]}}{\bf{H}}_{e}^{[K]}\right).
 \end{array}
\end{equation}

From the perturbation theory in [33], [34],  we formulate ${\bf{\bar v}}^{[1]}$  as
\begin{equation}
{{\bf{\bar v}}^{[1]}} =  {{{\bf{v}}^{[1]}} - \sqrt{\sigma _{M}}\sum\limits_{k = 1,k \ne m}^M {\frac{{{\bf{v}}_k^{\dag} \Delta{{\bf{H}}_{e}^{[K-1]}} {\bf{H}}_{e}^{[K]}{{\bf{v}}_m}}}{{{\lambda _m} - {\lambda _k}}}} {{\bf{v}}_k}}
\end{equation}
where ${{\bf{v}}^{[1]}}={{\bf{v}}_{m}}$ and $\{ {{\bf{v}}_1}, ... ,{{\bf{v}}_M}\}$ and $\{ \lambda_1, ... , \lambda_M\}$ are the set of eigenvectors and the corresponding eigenvalues of ${{\bf{H}}_{e}^{[K-1]}}{{\bf{H}}_{e}^{[K]}}$, respectively.

Therefore, the computation error of $\Delta {\bf \bar v}^{[1]}= {\bf v}^{[1]}-{\bf \bar v}^{[1]}$ is upper bounded by
\begin{equation} \label{effect_H}
{\mathbb E}\left[ {\left\| {\Delta {\bf{\bar v}}^{[1]} } \right\|^2 } \right] \le 2^{-{B_M \over {M^2-1}}}\underbrace {\sum\limits_{k = 1,k \ne m}^M  {\frac{{ {\left\| { {\bf{ H}}_e^{[K]} } \right\|^2 }}}{{ \left | ({\lambda _m} - {\lambda _k}) \right |^2 }}}}_{\bar a_1}.
\end{equation}
Applying the result in \eqref{effect_H} to Proposition \ref{Prop:H:E_SUM_I} , the residual interference due to the $\Delta {\bf \bar v}^{[1]}$ is upper-bounded by 
\begin{equation}
{\mathbb E}_{\bf H, \mathcal W}\left[ \sum\limits_{k=(K-1)}^{K}Pd_{k1}^{ - \alpha }{\left| {{{{\bf{\hat r}}}^{[k]\dag }}{{\bf{H}}^{[k1]}}{{{ \bf{ \bar v}}}^{[1]}}} \right|^2} \right]\le P\cdot 2^{-{{B_M}\over {M^2-1}}} \cdot \left( d_{(K-1)1}^{ - \alpha }+d_{K1}^{ - \alpha }\right)\cdot M^2 \cdot {\mathbb E}[\bar a_1]
\end{equation}
and its value should be constant  to achieve the full DoF. Therefore, the required feedback bits between receiver $(K-1)$ and $K$ is given by $B_M=(M^2-1)\cdot \log_2 P $.}
\bibliographystyle{ieeetr}

\clearpage

\end{document}